\documentclass[11pt, letter]{article}

\usepackage{booktabs}
\usepackage{latexsym}
\usepackage[centertags]{amsmath}
\usepackage{amsfonts}
\usepackage{amssymb}
\usepackage{amsthm}
\usepackage{newlfont}
\usepackage{graphicx}
\usepackage{float}
\usepackage{hhline}
\usepackage{algorithm}
\usepackage[noend]{algpseudocode}
\usepackage{fullpage}
\usepackage{amsmath}
\usepackage{mdframed}
\usepackage{verbatim}
\usepackage{multicol}
\usepackage{url}
\usepackage{enumitem}
\usepackage[margin=0.95in]{geometry}
\usepackage{wrapfig}
\usepackage{multirow}
\usepackage{array}
\usepackage{etoolbox,refcount}
\usepackage{multicol}
\usepackage[labelfont=bf]{caption}
\usepackage{epstopdf}
\usepackage{tabto}

\DeclareMathOperator*{\argmin}{argmin}

\makeatletter
\newcommand\footnoteref[1]{\protected@xdef\@thefnmark{\ref{#1}}\@footnotemark}
\makeatother

\newcolumntype{?}{!{\vrule width 2pt}}

\newcounter{countitems}
\newcounter{nextitemizecount}
\newcommand{\setupcountitems}{%
  \stepcounter{nextitemizecount}%
  \setcounter{countitems}{0}%
  \preto\item{\stepcounter{countitems}}%
}
\makeatletter

\newcommand{\computecountitems}{%
  \edef\@currentlabel{\number\c@countitems}%
  \label{countitems@\number\numexpr\value{nextitemizecount}-1\relax}%
}

\newcommand{\nextitemizecount}{%
  \getrefnumber{countitems@\number\c@nextitemizecount}%
}

\newcommand{\previtemizecount}{%
  \getrefnumber{countitems@\number\numexpr\value{nextitemizecount}-1\relax}%
}

\newenvironment{thmclone}[1]{\noindent
\textbf{Theorem~\ref{#1}} (restated).\em}{\par}

\newenvironment{lemclone}[1]{\noindent
\textbf{Lemma~\ref{#1}} (restated).\em}{\par}

\newenvironment{lemexe}[1]{\noindent
\textbf{Lemma~\ref{#1}} (extended).\em}{\par}

\makeatother    
\newenvironment{AutoMultiColItemize}{%
\ifnumcomp{\nextitemizecount}{>}{2}{\begin{multicols}{2}}{}%
  \setupcountitems\begin{enumerate}}%
{\end{enumerate}%
\unskip\computecountitems\ifnumcomp{\previtemizecount}{>}{2}{\end{multicols}}{}}

\newcommand{\tup}[2]{\ensuremath{\langle #1, #2 \rangle}}
\newcommand{\triple}[3]{\ensuremath{\langle #1, #2, #3 \rangle}}

\newcommand{\sasha}[1]{{\footnotesize\color{red}[\textbf{Sasha:}
#1]}}
\newfloat{algo}{htbp}{algo}
\floatname{algo}{Algorithm}

\usepackage{amssymb} 

\newtheorem{definition}{Definition}
\newtheorem{lemma}{Lemma}
\newtheorem{corollary}{Corollary}
\newtheorem{theorem}{Theorem}
\newtheorem{observation}{Observation}
\newtheorem{assumption}{Assumption}
\newtheorem{claim}{Claim}

\newcommand{\Write}{\ensuremath{\textsc{write}}}
\newcommand{\Read}{\ensuremath{\textsc{read}}}
\newcommand{\lwrite}{\ensuremath{\textit{write}}}
\newcommand{\lread}{\ensuremath{\textit{read}}}

\newcommand{\PntCont}{\ensuremath{\textrm{PntCont}}}
\newcommand{\writemax}{\ensuremath{\textit{write-max}}}
\newcommand{\readmax}{\ensuremath{\textit{read-max}}}
\newcommand{\maxreg}{\ensuremath{\textrm{max-register}}}
\newcommand{\maxregs}{\ensuremath{\textrm{max-registers}}}
\newcommand{\kreg}{\ensuremath{$k$\textrm{-register}}}
\newcommand{\relwrites}{\ensuremath{\textit{rel-writes}}}

\newcommand{\BlockedWrites}{\ensuremath{\text{BlockedWrites}}}

\algdef{SE}[OPERATION]{Operation}{EndOperation}[1]{\textbf{operation}\
#1\ }{\algorithmicend\ \textbf{}}%
\algtext*{EndEvent}

\algdef{SE}[MYFUNCTION]{Myfunction}{EndMyfunction}[1]{\textbf{function}\
#1\ }{\algorithmicend\ \textbf{}}%
\algtext*{EndOperation}

\algdef{SN}[MACRO]{Macro}{EndMacro}[1]{\ #1\ }{}%

\algdef{SN}[TYPES]{Types}{EndTypes}[1]{\hspace*{-0.2cm}\textbf{Types:}\
#1\ }{}%

\algdef{SN}[LOCAL]{Local}{EndLocal}[1]{\hspace*{-0.2cm}\textbf{Local
variables:}\ #1\ }{}%

\algdef{SN}[BASEOBJECTS]{BaseObjects}{EndBaseObjects}[1]{\hspace*{-0.2cm}\textbf{Base
    Objects:}\ #1\ }{}%
    
\algdef{SN}[BASEOBJECTSSERVERS]{BaseObjectsServers}{EndBaseObjectsServerss}[1]{\hspace*{-0.2cm}\textbf{Base
Objects and Servers:}\ #1\ }{}%
    
\algdef{SN}[SERVERS]{Servers}{EndServers}[1]{\hspace*{-0.2cm}\textbf{Servers:
 }\ #1\ }{}%
 
 \algdef{SN}[CLIENTSTATES]{States}{EndStates}[1]{\hspace*{-0.2cm}\textbf{Clients
 states:
 }\ #1\ }{}%

\algdef{SN}[CODE]{Code}{EndCode}[1]{\hspace*{-0.2cm}\textbf{Code}\
#1\ }{}%

\algdef{SE}[RRESPOND]{RRespond}{EndRRespond}[1]{\textbf{upon
receiving $b.read()$ respond}\ #1\
\algorithmicdo}{\algorithmicend\ \textbf{}}%
\algtext*{EndRRespond}

\algdef{SE}[WRESPOND]{WRespond}{EndWRespond}[1]{\textbf{upon
receiving $b.write(*)$ respond}\ #1\
\algorithmicdo}{\algorithmicend\ \textbf{}}%
\algtext*{EndWRespond}

\begin{document}


\title{Space Complexity of Fault Tolerant Register Emulations}

 \author{{Gregory Chockler} \\CS Department\\Royal
 Holloway, University of London, UK\\gregory.chockler@rhul.ac.uk\\
   \and
   {Alexander Spiegelman}\thanks{Alexander Spiegelman is
   grateful to the Azrieli Foundation for the award of an
   Azrieli Fellowship.} \\Viterbi EE
   Department\\Technion, Haifa, Israel\\sashas@tx.technion.ac.il
   }

\date{}
\maketitle

\begin{abstract}

  Driven by the rising popularity of cloud storage, the costs
  associated with implementing reliable storage services from a
  collection of fault-prone servers have recently become an actively
  studied question.  The well-known ABD result shows that an
  $f$-tolerant register can be emulated using a collection of $2f+1$
  fault-prone servers each storing a single read-modify-write object
  type, which is known to be optimal. In this paper we generalize this
  bound: we investigate the inherent space complexity of emulating
  reliable multi-writer registers as a fucntion of the type of the
  base objects exposed by the underlying servers, the number of
  writers to the emulated register, the number of available
  servers, and the failure threshold.

  We establish a sharp separation between registers, and both
  max-registers (the base object types assumed by ABD) and CAS in
  terms of the resources (i.e., the number of base objects of the
  respective types) required to support the emulation; we show that no
  such separation exists between max-registers and CAS.
  Our main technical contribution is lower and upper bounds on the
  resources required in case the underlying base objects are
  fault-prone read/write registers. We show that the number of
  required registers is directly proportional to the number of writers
  and inversely proportional to the number of servers.

\setcounter{page}{0}
\thispagestyle{empty}

\end{abstract}

\newpage

\section{Introduction}
\label{sec:intro}

Reliable storage emulations seek to construct
fault-tolerant shared objects, such as read/write
registers, using a collection of \emph{base objects} hosted on
failure-prone servers. 
Such emulations are core enablers for many modern storage
services and applications, including cloud-based online data
stores~\cite{pnuts,riak,spinnaker,zookeeper,mongodb} and
Storage-as-a-Service
offerings~\cite{s3,simple-db,dynamodb,azure-storage}.

Most existing storage emulation algorithms are
constructed from storage services capable of supporting
custom-built read-modify-write (RMW)
primitives~\cite{ABD95, EnglertS00, EnglertS00, rambo,
DuttaGLV10, GeorgiouNS09,
DBLP:journals/eatcs/AguileraKMMS10, mwr-journal, MR98}, and
perhaps the most famous one is ABD~\cite{ABD95}.
This algorithm emulates an $f$-tolerant atomic wait-free
register, accessed by an unbounded number of processes (readers
and writers), on top of $2f+1$ servers, each of which stores a
single RMW object.
Since $f$-tolerant register emulation is
impossible with less than $2f+1$ servers~\cite{attiya-quorum,
lynch-quorum}, the ABD algorithm's space complexity is optimal.

However, support for atomic RMW is not always available: the
operations exposed by network-attached disks are typically
limited to basic read/write capabilities, and the interfaces
exposed by cloud storage services sometimes augment this
with simple conditional update primitives similar to
Compare-and-Swap (CAS).
A natural question that arises is therefore how the
ABD results generalize when only weaker primitives (e.g.,
read/write registers) are available.
More specifically, we are interested whether reliable
storage emulations are possible with weaker primitives, and if
so, what their space complexity is, and in particular,
whether it depends on the number of writers and
the number of servers.


To answer these questions, we consider a collection of $n$ fault-prone
servers, each of which stores
base objects supporting the given primitive.
The failure granularity is servers,
meaning that a server crash causes all base objects it stores to
crash as well.
We study three primitives:
read/write register, max-register~\cite{aspnes2009max},
and CAS.
For each primitive, we are interested in the number of 
base objects required to emulate an $f$-tolerant register for
$k$ writers using $n$ servers.
  
%
%
To strengthen our result, we prove the lower bound under weak liveness
and safety guarantees, namely, obstruction freedom and write
sequential safety (WS-Safety).  The latter is a weak generalization
for Lamport's notion of safety~\cite{lamport1986interprocess} to
multi-writer registers, which we define in Section~\ref{sec:prelim}.
Since atomicity usually requires readers to write, which may induce a
dependency on the number of readers, we consider regularity for our
upper bound; we define in Section~\ref{sec:prelim} write sequential
regularity (WS-Regularity), which is a weaker form of multi-writer
regularity defined in~\cite{mwr-journal}.
The lower bound proved in~\cite{attiya-quorum,lynch-quorum} on the
number of servers required for $f$-tolerant register emulations can be
easily generalized for WS-Safe obstruction-free emulations.
Therefore, we assume that $n \geq 2f +1$ throughout the paper.

Table~\ref{table:hiererchy} summarizes 
our results.  Interestingly, even though both registers and
max-registers have the lowest
consensus number of $1$ in Herlihy's hierarchy~\cite{herlihy1991wait},
our results show they are clearly separated with respect to their
power to support a reliable multi-writer register in a
space-efficient fashion. On the other hand, no such separation exists
between CAS, which has an infinite consensus number, and max-register.
As an aside, we note that our classification has implications for the
standard shared memory model (without base object failures); for
example, it implies that a max-register for $k$ writers cannot be
implemented from less than $k$ read/write registers (proven in
Theorem~\ref{thm:maxreg}).

\begin{table*}[t]
\caption{The number of base objects used by
     $f$-tolerant register emulations with $k$ writers and
    $n>2f$ servers.}
\label{table:hiererchy}
\begin{minipage}{\textwidth}   
\centering
\begin{tabular}{ccc}
    \toprule
    	\multirow{2}{*}{\textbf{Base object}}  
%
    	
    	&Lower bound & Upper bound  \\
    	& (WS-Safe, obstruction-free) & (WS-Regular, 
    	wait-free)\\
    	
    \midrule
		
		max-register & $2f+1$ & $2f+1$ \\

		CAS & $2f+1$ & $2f+1$ \\

    
        	\multirow{2}{*}{read/write register} & 
    	\multirow{2}{*}{$ kf + \Big\lceil
    	\frac{k}{\frac{n-(f+1)}{f}} \Big\rceil (f+1)$}  &
    	\multirow{2}{*}{$kf +  \big\lceil \frac{k}{\big\lfloor
    	\frac{n-(f+1)}{f} \big\rfloor} \big\rceil (f+1)$}\\
    	
    	&&\\
    	
%

    
%


		
    \bottomrule

\end{tabular}
\end{minipage}
\end{table*}

\paragraph{Results.}
Despite the fact that the original ABD emulation~\cite{ABD95} assumes
a general RMW base object on every server, we observe that the code
executed by each server in the multi-writer ABD
protocol~\cite{rambo,mwr-journal,MR98} can be encapsulated into the
\writemax\/ (for handling update messages) and \readmax\/ (for
handling read messages) primitives of \maxreg\/.  Therefore, the upper
bound of $2f+1$ applies to \maxregs\/ as well.  
In Appendix~\ref{app:CAStoMax} we show how to emulate a
max-register from a single CAS in a wait-free manner. 
Thus, the upper bound for max-register also applies to
CAS.

Our main technical contribution is a lower bound on the number of
read/write registers required to emulate an $f$-tolerant WS-Safe
obstruction-free register for $k$ clients using $n$ servers.  While
the ABD~\cite{ABD95} space complexity does not depend on the number of
writers or the number of servers, we show in Section~\ref{sec:space}
(Theorem~\ref{thm:kf}) that when servers support only read/write
registers, the lower bound increases linearly with the number of
writers and decreases (up to a certain point) with the number of
available servers.  In particular, our lower bound implies that at
least $kf + f+1$ registers are needed regardless of the number of
available servers.  We exploit asynchrony to show that an emulated
write must complete even if it leaves $f$ pending writes on base
registers, forcing the next writer to use a different set of
registers, even in a write-sequential run.

In Section~\ref{sec:space}, we present a new upper bound construction 
that closely matches our lower bound (Theorem~\ref{thm:upper}). 
Note that the two bounds are closely aligned, and in particular,
coincide in the two important cases of $n = 2f+1 $ and
$n \geq kf +f +1$ where they are equal to $kf + k(f+1)$ and $kf+f+1$
respectively.  An interesting open question is to close the remaining
small gap.

Another open question
is whether our lower bound is tight for stronger regularity
definitions~\cite{shao2011multiwriter}.  In the special case of
$n=2f+1$ servers and $k$ writers, a matching upper bound of $(2f+1)k$
registers can be achieved by simply having each server implement a
single $k$-writer \maxreg\/ from $k$ base
registers~\cite{aspnes2009max}. The question of the general case of
$n\ge 2f+1$ remains open.

In Appendix~\ref{sec:space:app}, we show the following three
additional results implied by an extended variant of our main lower
bound construction: (1) a lower bound of $k$ registers per server for
$n = 2f+1$ (Theorem~\ref{thm:2fplus1}); (2) a lower bound on the
number of servers when the maximum number of registers stored on each
server is bounded by a known constant
(Theorem~\ref{thm:bounded-servers}); and (3) impossibility of
constructing fault-tolerant multi-writer register emulations adaptive
to point contention~\cite{point-cont-afek,point-cont}
(Theorem~\ref{theorem:adaptive}).

\paragraph{Related work.}
The space complexity of fault-tolerant register emulations has been
explored in a number of prior works. Aguilera et
al.~\cite{aguilera-nad} show that certain types of multi-writer
registers cannot be reliably emulated from a fixed number of
fault-prone ones if the number of writers is not a priori known. They
however, do not provide precise bounds on the number of base registers
as a function of the number of writers, the number of servers, and the
failure threshold as we do in this paper.
The space complexity of reliable register emulations in terms of the
amount of data stored on fault-prone RMW servers was studied
in~\cite{Fan2003,Chockler2007}, and more recently
in~\cite{spiegelman16,cadambe16,cadambe17}. Since we are only
interested in the number of stored registers and not their sizes,
these results are orthogonal to ours.

Basescu et al.~\cite{ICStore12} describe several fault-tolerant
multi-writer register emulations from a collection of fault-prone
read/write data stores. Their algorithms incorporate a 
garbage-collection mechanism that ensures that the storage cost
is adaptive to the write concurrency, provided that the
underlying servers can be accessed in a synchronous fashion. 
Our results show that asynchrony has a profound impact on
storage consumption
by exhibiting a {\em sequential}\/ failure-free
run where the number of registers that need to be stored
grows linearly with the number of writers.

The proof of our main result (see Lemma~\ref{lem:exhaustive-run})
further extends the adversarial framework of~\cite{spiegelman16} to
exploit the notion of {\em register covering}\/ (originally due
to~\cite{burns-lynch}) extended to fault-prone base registers as
in~\cite{aguilera-nad}.  Covering arguments have been successfully
applied to proving numerous space lower bound results in the
literature (see~\cite{faith-survey2014} for a survey) including the
recent tight bounds for obstruction-free
consensus~\cite{Gelashvili2015,Zhu:2016}, which are at
the heart of the space hierarchy
of~\cite{ellen2016complexity}.



\vspace{-3mm}
\section{Informal Model}
\label{sec:prelim}

In this section, we will introduce basic premises of our system
model in informal terms. 
The formal model can be found in Appendix~\ref{app:formal-model}. 

An {\em asynchronous fault-prone shared memory
system}~\cite{JCT98} consists of a set of shared {\em base}\/
objects ${\cal B} = \{b_1,b_2,\dots\}$ accessed by clients in
the set ${\mathbb C} = \{c_1, c_2, \dots\}$.
We extend the model by mapping objects to a set
${\cal S} = \{s_1, s_2, \dots\}$ of servers via a function
$\delta : {\cal B} \to {\cal S}$. We denote $n
\triangleq|\cal S|$. For $B \subseteq {\cal B}$, $\delta(B)$ 
denotes the {\em image}\/ of $B$, i.e., $\delta(B) = \{\delta(b)
: b \in B\}$.
Conversely, for $S \subseteq {\cal S}$, 
$\delta^{-1}(S) = \{b : \delta(b) \in S\}$.  
Both servers and clients can fail by crashing. 
A crash of a server cause all objects mapped to that server to
instantaneously crash\footnote{Note that the original fault-prone
shared model~\cite{JCT98} is a specific case of our model
when $\delta$ is an injective function.}.

 
We study algorithms emulating reliable {\em multi-writer/multi-reader
  (MWMR)}\/ registers to a set of clients. Our focus will be on a
register supporting an \emph{a priori fixed and known} set of $k>0$
writers, to which we will henceforth refer simply as
\emph{\kreg}. Clients interact with the emulated register via {\em
  high-level}\/ read and write operations. To distinguish the
high-level emulated reads and writes from low-level base object
invocations, we refer to the former as \Read\/ and \Write. We say that
high-level operations are {\em invoked}\/ and {\em return} whereas
low-level operations are {\em triggered}\/ and {\em respond}. A
high-level operation consists of a series of trigger and respond {\em
  actions}\/ on base objects, starting with the operation's invocation
and ending with its return. Since base objects are crash-prone, we
assume that the clients can trigger several low-level operations 
without waiting for the previously triggered operations to
respond.

An emulation algorithm $A$ defines the behaviour of clients as
deterministic state machines where state transitions are associated
with actions, such as trigger/response of low-level operations. A {\em
  configuration}\/ is a mapping to states from system components,
i.e., clients and base objects. An {\em initial configuration}\/ is
one where all components are in their initial states.

A {\em run}\/ of $A$ is a (finite or infinite) sequence of alternating
configurations and actions, beginning with some initial configuration,
such that configuration transitions occur according to $A$. We will
henceforth refer to an $A$'s transition occurring in a run (i.e., a
triple consisting of consecutive configuration, action, and
configuration) simply as a {\em step}.  We use the notion of
time $t$ during a run $r$ to refer to the configuration reached after $t$ steps
in $r$.
A run is {\em write-only}\/ if it has no invocations of
high-level \Read\/ operations. A run $r$ is {\em
write-sequential}\/ if no two high-level writes are concurrent
in~$r$.

We say that a base object, client, or server is {\em faulty}\/ in a
run $r$ if it fails at some time in $r$, and correct otherwise. A run
is {\em fair}\/ if (1) for every low-level operation triggered by a
correct client on a correct base object there is eventually a
matching response, and (2) every correct client gets infinitely many
opportunities to both trigger  low-level operations and execute
return actions. We say that a low-lever operation on a base object is
{\em pending}\/ in run $r$ if it was triggered but has no matching
response in $r$.
We assume that the base objects are atomic~\cite{herlihy} (see
Section~\ref{sec:consistency} of Appendix~\ref{app:formal-model}).

\vspace{-2mm}
\paragraph{Properties of emulation algorithms.}
Let $A$ be a fault-tolerant \kreg\/ emulation for some $k>0$.
We now give informal definitions of safety, liveness,
fault-tolerance, and space consumption properties of $A$. The formal
definitions can be found in Appendix~\ref{app:formal-model}.

\noindent {\em Write-Sequential Regularity (WS-Regular)}: $A$ is
write-sequential regular (WS-Regular) if for all its write-sequential
runs $r$, for each high-level read operation $Rd$, there exists a
linearization of the sub-sequence of $r$ consisting of $Rd$ and
all the high-level writes in $r$~\footnote{Note that this
definition is a generalization for Lamport's notion of
regularity to multi-writer register that coincide with it in
case of a single writer, but is weaker than multi-writer
regularity generalizations defined in~\cite{mwr-journal}.}.
  

\noindent {\em Write-Sequential Safety (WS-Safe)}: Similar to
WS-Regular, but only required to hold for high-level reads that
are not concurrent with any high-level writes.

\noindent {\em Wait Freedom}: $A$ is wait-free if it guarantees that
every high-level operation invoked by a correct client eventually
returns.

\noindent {\em Obstruction Freedom}: $A$ is obstruction-free
if every high-level operation invoked by a correct client that is not
concurrent with any other high-level operation by a correct client
eventually returns.

\noindent {\em $f$-tolerance}: $A$ is $f$-tolerant if it remains
correct (in the sense of its safety and liveness properties) as long
as at most $f$ servers crash for a fixed $f > 0$.

\noindent {\em Resource Complexity}: The {\em resource consumption}\/ of $A$
in a run $r$ is the number of base objects used by $A$ in $r$. The
{\em resource complexity}~\cite{JCT98} of $A$ is the maximum resource
consumption of $A$ in all its runs.

\vspace{-3mm}
\section{Resource Complexity of Write-Sequential
\kreg\ Emulation}
\label{sec:space}

In Section~\ref{sub:LowerBoundOverview} we give an overview and
intuition for our lower bound, and in
Section~\ref{sub:LowerBound} we prove it.
In Section~\ref{sub:UpperBound} we present a closely matching
upper bound algorithm.


\vspace{-2mm}
\subsection{Lower bound overview}
\label{sub:LowerBoundOverview}

We prove that any $f$-tolerant emulation of an obstruction-free
WS-Safe \kreg\/ from of a collection of MWMR atomic registers stored
on a collection $\mathcal{S}$ of crash-prone servers has resource
complexity of at least $ kf + \Big\lceil \frac{kf}{|\mathcal{S}|-(f+1)}
\Big\rceil (f+1)$.



Our proof exploits the fact that the environment is allowed to
prevent a pending low-level write from taking effect
for arbitrarily long~\cite{aguilera-nad}. As a result, a client
executing a high-level \Write\ operation cannot reliably store
the requested value in a base register that has a pending write
as this write may take effect at a later time thus erasing the stored value.
At the same time, the client cannot wait for all base registers
on which it triggers low-level operations to respond, since up to
$f$ of them may reside on faulty servers.
It therefore must be able to complete a high-level write without
waiting for responses from up to $f$ registers. Consequently, the next high-level write (by a different
client) cannot reliably use these registers (as they might have
outstanding low-level writes), and is therefore forced to use
additional registers thus causing the total number of registers grow
with each subsequent write.

In our main lemma (Lemma~\ref{lem:exhaustive-run}), we formalize this
intuition as follows: Starting from a run $r_0$ consisting of an
initial configuration, we build a sequence of consecutive extensions
$r_1,\dots,r_k$ so that $r_i$ is obtained from $r_{i-1}$ by having a
new client invoke a high-level write $W_i$ of some (not
previously used) value. We then let the environment behave in an
adversarial fashion (Definition~\ref{def:ad}) by blocking the
responses from the writes triggered on at most $f$ base
registers as well as the prior pending low-level writes.  In
Lemma~\ref{lem:1write}, we show that $W_i$ must terminate
without waiting for these responses to arrive. Furthermore, in
Lemma~\ref{lem:2f}, we show that $W_i$ must invoke low-level
writes on at least $2f+1$ base registers (residing on $\ge
2f+1$ different servers) that do not have any prior pending
writes.
This, combined with Lemma~\ref{lem:1write}, implies that by the
time $W_i$ completes, there are at least $f$ more registers on
at least $f$ servers with pending writes after $W_i$ completes.
Thus, by the time the $k$th high-level write completes, the total number of covered registers is at least $kf$ (see
Lemma~\ref{lem:exhaustive-run}(\ref{lem:CovF}).

To obtain a stronger bound, our construction is parameterized by
an {\em arbitrary}\/ subset $F$ of servers such that $|F|=f+1$.
We show that the extra storage available on these servers cannot
in fact, be utilized by an emulation (see
Lemma~\ref{lem:exhaustive-run}(\ref{lem:F}) forcing it to use at
least $kf$ registers on the remaining $\mathcal{S} \setminus F$
servers to accommodate the same number $k$ of writers.  We use
this result in Theorem~\ref{thm:kf}, to show that the number of
base registers required for the emulation is a function of $k$
and $|\mathcal{S}|$.

\vspace{-2mm}
\subsection{Lower Bounds}
\label{sub:LowerBound}

%



For any time $t$ (following the $t$\textsuperscript{th} action) in a
run $r$ of the emulation algorithm we define the following:

\begin{itemize}

\item {\em Covering write}: Let $w$ be a low-level write triggered on
  a base register $b$ at times $\le t$. We will refer to $w$ as
  {\em covering}\/ \emph{at} $t$, and to $b$ as being {\em
  covered by}\/ $w$ \emph{at} $t$ if $w$ is pending at $t$.

  
\item $C(t) \subseteq \mathbb{C}$: the set of all clients that
have completed a high-level write operation at times $\le t$.

%

  \item $Cov(t) \subseteq \mathcal{B}$: the set of all base registers
  being covered by some low-level write at time $t$.



\end{itemize}



\noindent
We first prove the following key lemma:

\begin{lemma}
  For all $k>0$, $f>0$, let $A$ be an $f$-tolerant algorithm that
  emulates a WS-Safe obstruction-free \kreg\/ using a collection
  $\mathcal{S}$ of servers storing a collection $\mathcal{B}$ of
  wait-free MWMR atomic registers.  Then, for every $F \subseteq S$
  such that $|F| = f +1$, there exist $k$ failure-free runs $r_i$,
  $0 \le i \le k$, of $A$ such that (1) $r_0$ is a run consisting of
  an initial configuration and $t_0=0$ steps, and (2) for all
  $i \in [k]$, $r_i$ is a write-only
  sequential extension of $r_{i-1}$
  ending at time $t_i
  > 0$ that consists of $i$ complete high-level writes of
  $i$
  distinct values $v_1,\dots,v_i$
  by $i$ distinct clients $c_1,\dots,c_i$ such that:

  
  \begin{enumerate}[label=\alph*)]
  \begin{multicols}{2}
    
  \item $|Cov(t_i)| \ge if$
    \label{lem:CovF}
    
  \item $\delta(Cov(t_i)) \cap F = \emptyset$
    \label{lem:F}
    
%
%

  \end{multicols}     
  \end{enumerate}

  \label{lem:exhaustive-run}
\end{lemma}

\vspace*{-3mm}
Fix arbitrary $k>0$, $f>0$, and a set $F$ of servers such that
$|F|=f+1$. We proceed by induction on $i$, $0 \le i \le k$.
\textbf{Base:} Trivially holds for the run $r_0$ of $A$ consisting of
$t_0=0$ steps.  \textbf{Step:} Assume that $r_{i-1}$ exists for all
$i\in [k-1]$.  We show how $r_{i-1}$ can be extended up to time
$t_i > t_{i-1}$ so that the lemma holds for the resulting run $r_i$.
%
For the remainder of the Lemma~\ref{lem:exhaustive-run}'s proof, we
will assume without loss of generality that every low-level write
operation is linearized simultaneously with its respond step.
Formally:
\begin{assumption}[Write Linearization]
  For every extension $r$ of $r_{i-1}$ and a base object
  $b\in \mathcal{B}$, let $L_{r|b}$ be a linearization of $r |
  b$. Then, $L_{r|b}$ does not include any low-level write %
  operations in $pending(r | b)$, and for any two low-level writes
  $w_1, w_2 \in complete(r | b)$ such that
  $respond(w_1) \prec_{r | b} respond(w_2)$,
  $w_1 \prec_{L_{r | b}} w_2$.
\label{ass:linear}
\end{assumption}
Note that the above implies that no low-level write $w$ that is
covering some register $b$ at a time $t$ in $r$ will be observed
by any low-level reads from $b$ as having taken effect until
after $w$'s respond event occurs.

We proceed by introducing the following notation:
\begin{definition}
  Let $r$ be an extension of $r_{i-1}$. For all times $t\ge t_{i-1}$
  in $r$, let

\begin{enumerate}

  
\item $Tr_i(t)$: the set of all base
  registers which have a low-level write triggered on between
  $t_{i-1}$ and~$t$. \label{def:notation:tr}

\vspace*{-5mm}
\item $Rr_i(t) \subseteq Tr_i(t)$ be the set of all base registers
  which had a low-level write triggered on and responded (took effect)
  between $t_{i-1}$ and $t$. \label{def:notation:rr}
  
\vspace*{-1mm}
\item $Cov_i(t) = Cov(t) \setminus Cov(t_{i-1})$ be the set of all
  base registers that have been newly covered between $t_{i-1}$ and
  $t$. Note that $Cov_i(t) \subseteq
  Tr_i(t)$. \label{def:notation:covi}

 \vspace*{-1mm}
\item $Q_i(t) \subseteq {\cal S}$ be the set of all servers such that
  $Q_i(t) = \delta(Cov_i(t)) \setminus F$ if
  $|\delta(Cov_i(t)) \setminus F| \le f$, and $Q_i(t) = Q_i(t-1)$,
  otherwise. \label{def:notation:qi}

\vspace*{-1mm}
\item
  $F_i(t) \triangleq \{ s \in F \mid \delta^{-1}(\{s\}) \cap Rr_i(t)
  \neq \emptyset \}$,
  i.e., $F_i(t)$ is the set of all servers in $F$ having a register
  that responded to a low-level write invoked after
  $t_{i-1}$. \label{def:notation:fi}

\vspace*{-1mm}
\item $M_i(t) \triangleq \delta(Cov_i(t)) \cap (F \setminus F_i(t))$,
  i.e., $M_i(t)$ is the set of all servers in $F$ with at least one
  register covered by a low-level write invoked after $t_{i-1}$ and
  without registers that have responded to the low-level
  writes invoked after $t_{i-1}$.
  \label{def:notation:mi}

\item $G_i(t) \subseteq {\cal S}$ be the set of all servers such that
  $G_i(t) = M_i(t)$ if $|Q_i(t)|<|F_i(t)|$, and $G_i(t) = \emptyset$,
  otherwise. \label{def:notation:gi}

\end{enumerate}
\label{def:notation}
\end{definition}

\noindent
Below we will introduce the adversary $Ad_i$, which causes $A$ to
gradually increase the number of base registers covered after
$t_{i-1}$ by delaying the respond actions of some of the previously
triggered low-level writes.

\begin{definition}[Blocked Writes] Let $r$ be an extension of
$r_{i-1}$. For all times $t\ge t_{i-1}$ in $r$, let
$\BlockedWrites_i(t)$ be the set of all low-level covering
writes $w$ satisfying either one of the following two
conditions:
  \begin{enumerate}

    \item $w$ was triggered by a client in $C(t_{i-1})$, or

    \item $w$ was triggered on a base register in
      $\delta^{-1}(Q_i(t) \cup G_i(t))$.

  \end{enumerate}
\label{def:blocked-writes}
\end{definition}

We say that a pending low-level write $w$ is \emph{blocked} in an
extension $r$ of $r_{i-1}$ if there exists a time $t\ge t_{i-1}$ such
that for all $t'>t$ in $r$, $op \in \BlockedWrites_i(t')$.  The
following definition specifies the set of the environment behaviours
that are allowed by $Ad_i$ in all extensions of $r_{i-1}$:







\begin{definition}[$Ad_i$]
	
  For an extension $r$ of $r_{i-1}$ we say that the environment {\em
    behaves like $Ad_i$}\/ after $r_{i-1}$ in $r$ if the following
  holds:
	 
    \begin{enumerate}
      
    \item There are no failures after time $t_{i-1}$ in $r$.
    
    \item \label{def:ad:block} For all $t \geq t_{i-1}$ in $r$, if a low-level write
      $w \in \BlockedWrites_i(t)$, then $w$ does not respond at
      $t$. 
    
    \item If $r$ is infinite then:
    
    \begin{enumerate}
    \item Every pending low-level read or write that is not blocked in
      $r$ eventually responds.

    \item Every trigger or return action that is ready to be executed
      at a client $c$ in $r$ is eventually executed.
     \end{enumerate}

    \end{enumerate}

\label{def:ad}
\end{definition}

\vspace*{-3mm}
For a finite extension $r$ of $r_{i-1}$, we will write $\tup{r}{Ad_i}$
to denote the set of all extensions of $r$ in which the environment
behaves like $Ad_i$ after $r_{i-1}$; and we will write
$\triple{r}{Ad_i}{t}$ to denote the subset of $\tup{r}{Ad_i}$
consisting of all runs having exactly $t$ steps.
For $X\in \{Q_i, F_i, M_i\}$ and a run
$r\in \triple{r_{i-1}}{Ad_i}{t}$, we say that $X$ is {\em stable}\/
after $r$ if for all $t' \geq t$ for all extensions $r'\in
\triple{r}{Ad_i}{t'}$, $X(t') = X(t)$.

The following lemma (proven in Section~\ref{sec:space:app} of the
appendix) asserts several technical facts implied directly by
Definitions~\ref{def:notation} and \ref{def:ad}.

\begin{lemma}
For all $t \ge t_{i-1}$ and $r\in \triple{r_{i-1}}{Ad_i}{t}$, all of
the following holds at time $t$ in $r$:
\begin{AutoMultiColItemize}
\item $Q_i(t) \subseteq \delta(Cov_i(t)) \setminus F$ \label{lem:facts:qsubset}
\item $Q_i(t) \subseteq Q_i(t+1)$  \label{obs:ad}
\item $F_i(t) \subseteq F_i(t+1)$  \label{obs:adf}
\item $|F_i(t)| - |Q_i(t)| \leq 1$ \label{obs:ad1}
\item $|Q_i(t)| \le f$ \label{obs:qlim}
\item $|F_i(t)| \le f+1$ \label{obs:filim}
\item $F_i(t) = F_i(t+1)\ \implies M_i(t) \subseteq M_i(t+1)$
 \label{lem:facts:covfi}
\item $|M_i(t)| \le f+1$ \label{lem:facts:maxcovfi}
\item $|\delta(Cov_i(t)) \setminus F| \ge f \implies |Q_i(t)|\ge f$
\label{lem:facts:qisize}
\item $|\delta(Cov_i(t)) \setminus F| < f \implies \delta(Rr_i(t)) \setminus F = \emptyset$
\label{lem:facts:rrr}
\item $(Q_i(t) \cup M_i(t)) \cap \delta^{-1}(Rr_i(t)) = \emptyset$
\label{lem:facts:silent}
\end{AutoMultiColItemize}
\label{lem:facts1}
%
\end{lemma}

\vspace{-2mm}
 The following corollary follows immediately from the
claims \ref{obs:ad}--\ref{obs:adf} and
\ref{obs:qlim}--\ref{lem:facts:maxcovfi} of
Lemma~\ref{lem:facts1}.

\begin{corollary}
  \hspace{-0.2cm} There exists a run $r \in \tup{r_{i-1}}{Ad_i}$ such that $F_i$,
  $Q_i$, and $M_i$ are all stable after $r$.
  %
\label{cor:stable}
\end{corollary}

We first show that $r_{i-1}$ can be extended with a complete
high-level write $W_i$ by a new client $c_i$ such that the environment
behaves like $Ad_i$ until $W_i$ returns. Roughly, the reason for this
is that $Ad_i$ guarantees that after $r_{i-1}$, $c_i$ would only miss
responses from for the writes invoked on at most $f$ servers (see
Claim~\ref{claim:size}) as well as those that might have been invoked
in $r_{i-1}$ by other clients $\{c_1,\dots,c_{i-1}\}$, which $c_i$ is
unaware of.  As a result, the involved servers and clients would
appear to $c_i$ as faulty after $r_{i-1}$, and therefore, to ensure
obstruction freedom, it must complete $W_i$ without waiting for the
outstanding writes to respond.

\begin{lemma}
  Let $r \in \tup{r_{i-1}}{Ad_i}$ be a run consisting of $r_{i-1}$
  followed by a high-level write invocation $W_i$ by client
  $c_i \not\in C(t_{i-1})$. Then, there exists a run
  $r_r \in \tup{r}{Ad_i}$ in which $W_i$ returns.
  \label{lem:1write}
\end{lemma}

By Corollary~\ref{cor:stable}, there exists an extension
$r' \in \triple{r}{Ad_i}{t'}$ where $t' > t_{i-1}$ such that $Q_i$,
$F_i$, and $M_i$ are all stable after $r'$.  If $W_i$ returns in $r'$,
we are done. Otherwise, we will first bound the number of servers
$|Q_i(t') \cup M_i(t')|$ controlled by $Ad_i$ as follows:

\begin{claim}
 Consider a time $t > t_{i-1} $, and a run $r \in
  \triple{r_{i-1}}{Ad_i}{t}$. If $M_i$ is stable after $r$, then
  $|Q_i(t) \cup M_i(t)| \le f$.
\label{claim:size}
\end{claim}

\begin{proof}
  By Lemma~\ref{lem:facts1}.\ref{obs:qlim}, $|Q_i(t)| \leq f$.  Thus,
  if $M_i(t) = \emptyset$, then $|Q_i(t) \cup M_i(t)| \le f$ as
  needed.  Otherwise ($M_i(t) \neq \emptyset$), we show that
  $|F_i(t)| = |Q_i(t)| + 1$.  Suppose to the contrary that
  $|F_i(t)| \neq |Q_i(t)| + 1$.  Since by
  Lemma~\ref{lem:facts1}.\ref{obs:ad1}, $|F_i(t)| \le |Q_i(t)| + 1$,
  the only possibility for $|F_i(t)| \neq |Q_i(t)| + 1$ is if
  $|F_i(t)| \le |Q_i(t)|$. Thus, by
  Definition~\ref{def:notation}.\ref{def:notation:gi},
  $G_i(t) = \emptyset$, and hence, by
  Definition~\ref{def:blocked-writes}, no writes on the registers in
  $\delta^{-1}(M_i(t))$ are blocked. However, since
  $M_i(t) \neq \emptyset$, at least one register in
  $\delta^{-1}(M_i(t))$ must have an outstanding write.  Therefore, by
  Definition~\ref{def:ad}.3(a), there exists time $t'$, and an
  extension $r'\in \triple{r}{Ad_i}{t'}$ such that one of the
  registers on some server $s \in M_i(t)$ responds at time $t'$.
  Thus, $s \in F_i(t')$, and therefore, $s\not\in M_i(t')$.  Hence,
  $M_i$ is not stable after $r$.  A contradiction to the assumption.
  
%

  Since $F_i(t) \subseteq F$,
  $ |F\setminus F_i(t)| + |F_i(t)| = |F| = f+1$. 
  Hence, $|F\setminus F_i(t)| = f + 1 - |F_i(t)| = f - |Q_i(t)|$.
  Since by Definition~\ref{def:notation}.\ref{def:notation:mi},
  $M_i(t) \subseteq (F\setminus F_i(t))$,
  $|M_i(t)| \le |F\setminus F_i(t)| = f - |Q_i(t)|$. Thus, we receive
  $|Q_i(t) \cup M_i(t)| \le |Q_i(t)| + |M_i(t)| \le f$ as needed.
\end{proof}






\noindent 
We are now ready to complete the proof of Lemma~\ref{lem:1write}:

\begin{proof}[Proof of Lemma~\ref{lem:1write}]
  By Claim~\ref{claim:size}, $|Q_i(t') \cup M_i(t')| \le f$, and by
  Lemma~\ref{lem:facts1}.\ref{lem:facts:silent}, no base registers on
  servers in $Q_i(t') \cup M_i(t')$ have ever responded to any
  low-level writes issued after $t_{i-1}$.  Thus, there exists a
  finite run $r''$, which is identical to $r'$ except all servers in
  $Q_i(t') \cup M_i(t')$ crash immediately after $r$ and each client
  $c_1,\dots,c_{i-1}$ fails before any of its covering writes on
  registers in $Cov(t_{i-1})$ responds.
  By $f$-tolerance and obstruction freedom, there exists a fair
  extension $r''\sigma$ of $r''$ (i.e.,
  $r''\sigma \notin \tup{r''}{Ad_i}$) such that $W_i$ returns in
  $r''\sigma$.  Since $Q_i \cup M_i$ is stable after $r'$, the set of
  registers precluded by $Ad_i$ from responding in $r'$ is identical
  to that in $r''$, and by Assumption~\ref{ass:linear}, no write with
  a missing response is linearized, $r'$ is indistinguishable from
  $r''$ to $c_i$. Thus, $r' \sigma\in \tup{r}{Ad_i}$, and since
  $\sigma$ includes the return event of $W_i$, $r'\sigma$ satisfies
  the lemma.
  
\end{proof}

\noindent We next show that in order to guarantee safety in the
face of the environment behaving like $Ad_i$, $W_i$ must trigger
a low-level write on at least one non-covered register on each server in a set of
$2f+1$ servers. 
An illustration of the runs constructed in the proof appears in
Figure~\ref{fig:runs} in Appendix~\ref{sec:space:app}.

\begin{lemma}
  Consider a run $r \in \triple{r_{i-1}}{Ad_i}{t_r}$ where
  $t_r > t_{i-1}$, consisting of $r_{i-1}$ followed by a complete
  high-level write invocation by client $c_i \not\in C(t_{i-1})$ that
  returns at time $t_r$. Then,
  $|\delta(Tr_i(t_r) \setminus Cov(t_{i-1}))| > 2f$.
\label{lem:2f}
\end{lemma}

\begin{proof}
  Denote $X \triangleq \delta(Tr_i(t_r) \setminus Cov(t_{i-1}))$, and
  assume by contradiction that $|X| \leq 2f$.  Let
  $S_1= F_i(t_r)$, $S_2=Q_i(t_r)$,
  $S_3 = X \cap (F \setminus F_i(t_r))$ and
  $S_4=X \setminus (S_1\cup S_2 \cup S_3)$.  Note that
  $S_1,S_2,S_3,S_4$ are pairwise disjoint, and
  $X = S_1 \cup S_2 \cup S_3 \cup S_4$.  

  We first show that $|S_1 + S_4| \le f$. By
  Lemma~\ref{lem:facts1}.\ref{obs:filim}, $|S_1|\le f+1$. 
  However, if $|S_1|=f+1$, then by Lemma~\ref{lem:facts1}.\ref{obs:ad1},
  $|S_1| - |S_2| = f + 1 - |S_2| \le 1$, and therefore,
  $|S_1| + |S_2|\ge 2f+1$ violating the assumption. Hence,
  $|S_1| \le f$. By Lemma~\ref{lem:facts1}.\ref{obs:qlim},
  $|S_2|\le f$. If $|S_2| = f$, then by assumption,
  $|S_1 \cup S_3 \cup S_4| = |S_1|+|S_3|+|S_4| \leq f$, and
  therefore, $|S_1+S_4| = |S_1| + |S_4| \le f$. And if $|S_2|<
  f$, then by Definitions~\ref{def:notation}.\ref{def:notation:qi} and
  \ref{def:ad}, 
  $|S_4|=0$.
  Hence, $|S_1 + S_4| = |S_1| + |S_4| \le f$.
  
  The proof proceeds by applying the partitioning argument to the sets
  $S_i$, $i\in [4]$ (see Appendix~\ref{sec:space:app}).

\end{proof}

\vspace{-2mm}
\noindent 
The following corollary follows immediately from Lemma~\ref{lem:2f},
Definitions~\ref{def:notation}.\ref{def:notation:qi} and \ref{def:ad},
and the choice of $|F|=f+1$:

\begin{corollary}
  Consider a run $r \in \triple{r_{i-1}}{Ad_i}{t_r}$ where
  $t_r > t_{i-1}$, consisting of $r_{i-1}$ followed by a complete
  high-level write invocation by client $c_i \not\in C(t_{i-1})$ that
  returns at time $t_r$. Then, $|Q_i(t_r)| = f$.
\label{cor:Q}
\end{corollary}

\vspace{-2mm}
\noindent
We are now ready to complete the proof of the induction step of
Lemma~\ref{lem:exhaustive-run}:

\begin{proof}[Proof of the induction step
(Lemma~\ref{lem:exhaustive-run})]
  
  By Lemma~\ref{lem:1write}, there exists a run
  $r \in \triple{r_{i-1}}{Ad_i}{t_r}$, $t_r > t_{i-1}$, in which
  $r_{i-1}$ is followed by a complete high-level write invocation
  $W_i$ by client $c_i \neq c_{i-1}$ writing a value
  $v_i \neq v_{i-1}$ and returning at time $t_r$.
  By Corollary~\ref{cor:Q}, $|Q_i(t_r)|=f$, and therefore, by
  Lemma~\ref{lem:facts1}.\ref{obs:ad1}, $|F_i(t_r)|=f+1$. Since
  $F_i(t_r) \subseteq F$ and $|F|=f+1$, we conclude that
  $F_i(t_r) = F$. Hence, by
  Definition~\ref{def:notation}.\ref{def:notation:mi},
  $M_i(t_r) = \emptyset$, which by
  Definition~\ref{def:notation}.\ref{def:notation:gi}, implies that
  $G_i=\emptyset$. Thus, by Definition~\ref{def:blocked-writes}, no
  writes on the registers in $\delta^{-1}(F)$ triggered after $t_{i-1}$
  are blocked.

  Hence, by Definition~\ref{def:ad}.3(a), there exists an extension
  $r' \in \triple{r}{Ad_i}{t'}$, for some $t' \ge t_r$, such that
  $\delta(Cov_i(t')) \cap F = \emptyset$.  We now show that $r_i = r'$
  and $t_i = t'$ satisfy the lemma.  By the induction hypothesis and
  the construction of extension $r'$, $r'$ is a write-only
  failure-free sequential extension of $r_{i-1}$ ending at time $t'$
  that consists of $i$ complete high-level writes of values
  $v_1,\dots,v_i$ by $i$ distinct clients $c_1,\dots,c_i$.  It remains
  to show that the implications~(\ref{lem:CovF}--\ref{lem:CovSup}
  hold for $t_i=t'$:

  \begin{enumerate}[label={}]
    
  \item \ref{lem:CovF} $|Cov(t')| \ge if$: By the induction hypothesis
    $|Cov(t_{i-1})| \ge (i-1)f$, and by Definition~\ref{def:ad},
    $Cov(t_{i-1}) \subseteq Cov(t')$.
    Therefore, we left to show that $|Cov(t') \setminus
    Cov(t_{i-1})| \geq f$.
    Since by Corollary~\ref{cor:Q}, $|Q_i(t_r)|=f$, and by
    Lemma~\ref{lem:facts1}.\ref{obs:ad}, $Q_i(t_r) \subseteq
    Q_i(t')$, we get $|Q_i(t')|=f$.
    By Definition~\ref{def:notation}.\ref{def:notation:qi},
    $|Cov_i(t')| \geq |Q_i(t')| $, and by
    Definition~\ref{def:notation}.\ref{def:notation:covi},
    $|Cov(t') \setminus Cov(t_{i-1})| =
    |Cov_i(t')|$.
    Therefore, we get $|Cov(t') \setminus Cov(t_{i-1})| \geq f$.  
    
  \item \ref{lem:F} $\delta(Cov(t')) \cap F = \emptyset$: By the
    induction hypothesis we get that
    $\delta(Cov(t_{i-1})) \cap F = \emptyset$, and by construction of
    $r'$ we get that $\delta(Cov_i(t')) \cap F = \emptyset$.  By
    Definition~\ref{def:notation}.\ref{def:notation:covi},
    $Cov(t') = Cov_i(t') \cup Cov(t_{i-1})$.  Therefore,
    $\delta(Cov(t')) \cap F = \emptyset$.

    
  \end{enumerate} 
\end{proof}

\vspace*{-7mm}
\noindent
\paragraph{Resource Complexity.}
We will now use Lemma~\ref{lem:exhaustive-run} to characterize the
minimum resource complexity of the algorithms implementing an
$f$-tolerant obstruction-free WS-Safe \kreg\/ as a function of the
number $|\mathcal{S}|$ of available servers. First, it is easy to see
that if $|\mathcal{S}| \le 2f$, then no such algorithm can exist. This
result is implied by an extended statement of
Lemma~\ref{lem:exhaustive-run} (see Theorem~\ref{thm:2f+1} in
Appendix~\ref{sec:space:app}), and can also be shown directly by a
straightforward application of a partitioning argument as discussed
in~\cite{lynch-quorum,attiya-quorum}.  If $|\mathcal{S}| > 2f$, then
we have the following:

\begin{theorem}
\label{thm:kf}
  For all $k>0$, $f>0$, let $A$ be an $f$-tolerant
  algorithm emulating an obstruction-free WS-Safe \kreg\/
  using a collection $\mathcal{S}$ of servers such that
  $|\mathcal{S}| \geq 2f + 1 $. 
  Then, $A$ uses at least $kf + \Big\lceil
  \frac{kf}{|\mathcal{S}|-(f+1)} \Big\rceil (f+1)$ base registers
  (i.e., $|\delta^{-1}(\mathcal{S})| \geq kf + \Big\lceil
  \frac{kf}{|\mathcal{S}|-(f+1)} \Big\rceil (f+1)$).
\end{theorem}

\begin{proof}

  Let $G \subseteq \mathcal{S}$ be the set consisting of all servers
  that store at least
  $\Big\lceil \frac{kf}{|\mathcal{S}|-(f+1)} \Big\rceil$ base
  registers (i.e.,
  $\forall s \in G, ~ |\delta^{-1}(\{s\})| \geq \Big\lceil
  \frac{kf}{|\mathcal{S}|-(f+1)} \Big\rceil$
  and
  $\forall s \in \mathcal{S} \setminus G, |\delta^{-1}(\{s\})| <
  \Big\lceil \frac{kf}{|\mathcal{S}|-(f+1)} \Big\rceil$).
  We first show that $|G| \geq f+1$.  Suppose toward a contradiction
  that $|G| < f+1$, and pick a set $F$, such that $|F|=f+1$ and
  $\mathcal{S} \supset F \supset G$.  By
  Lemma~\ref{lem:exhaustive-run}.\ref{lem:CovF}-\ref{lem:F}, there
  exists a run $r$ of $A$ consisting of $t$ steps such that
  $|Cov(t)| \ge kf$ and $\delta(Cov(t)) \cap F = \emptyset$.  Thus, by
  the pigeonhole argument, and since $|S\setminus F|= |S| -(f+1)$,
  there is at least one server in $\mathcal{S} \setminus F$ that
  stores at least
  $\Big\lceil \frac{kf}{|\mathcal{S}|-(f+1)} \Big\rceil$.  Therefore,
  since $F \supset G$ and the number of objects stored on a server is
  an integer, we get that there is at least one server in
  $\mathcal{S} \setminus G$ that stores at least
  $\Big\lceil \frac{kf}{|\mathcal{S}|-(f+1)} \Big\rceil$ base
  registers.  A contradiction.

  We get that
  $|\delta^{-1}(G)| \geq \Big\lceil \frac{kf}{|\mathcal{S}|-(f+1)}
  \Big\rceil (f+1)$.
  Now, again by
  Lemma~\ref{lem:exhaustive-run}.\ref{lem:CovF}-\ref{lem:F}, there
  exists a run $r'$ of $A$ consisting of $t'$ steps such that
  $|Cov(t')| \ge kf$ and $\delta(Cov(t')) \cap G = \emptyset$, meaning
  that $|\delta^{-1}(\mathcal{S} \setminus G)| \geq kf$.  Therefore,
  we get that
  $|\delta^{-1}(\mathcal{S})| \geq kf + \Big\lceil
  \frac{kf}{|\mathcal{S}|-(f+1)} \Big\rceil (f+1)$.
\end{proof}

The following bound on the number of registers required to emulate a
single (i.e., non-fault-tolerant) \maxreg\/ is a direct consequence of
Theorem~\ref{thm:kf} (see Appendix~\ref{sec:space:app} for the proof):

\begin{theorem}[Resource Complexity of $k$-\maxreg]
  For all $k>0$, any algorithm implementing a wait-free $k$-writer
  \maxreg\/ from a collection of wait-free MWMR atomic base registers
  uses at least $k$ base registers.
\label{thm:maxreg}
\end{theorem}

In Appendix~\ref{sec:space:app}, we prove an extended statement of
Lemma~\ref{lem:exhaustive-run}, and use it to show three additional
lower bounds as discussed in Section~\ref{sec:intro}.

\subsection{Upper Bound}
\label{sub:UpperBound}

In this section we present an $f$-tolerant construction
emulating a {\em wait-free WS-Regular}\/ \kreg\/ for all
combinations of values of the parameters $k>0$, $f>0$,
and $n$ where $n>2f$. 
Our construction is carefully crafted to deal with the adversarial
behaviour (Definition~\ref{def:ad}) that was exploited in the proof
of Lemma~\ref{lem:exhaustive-run} while minimizing the resource
complexity.
Similarly to multi-writer ABD~\cite{rambo,mwr-journal,MR98}, our
algorithm uses {\em read}\/ and {\em write quorums}\/ to read from and
write to registers.  However, since RMW objects are replaced with
read/write registers, and covering low-level writes belonging to old
\Write s can overwrite registers at any time, the quorums in our case
must have a larger intersection.

Let $z \triangleq \big\lfloor \frac{n- (f+1)}{f} \big\rfloor$ and
$y \triangleq zf + f +1$, we construct a collection $\mathcal{R}$ of
$ m = \big\lfloor \frac{k}{z} \big\rfloor$ disjoint sets
$R_0,\ldots,R_{m-1}$, each of which consist of $y$ registers, and
if $k/z$ is not an integer, then we add to $\mathcal{R}$ another
disjoint set $R_{m}$ of
$(k -\big\lfloor \frac{k}{z} \big\rfloor z) f + f +1$
registers. Intuitively, $z$ is the maximum number of writers
that can be supported by a single set of $y$ registers as can be
deduced from Lemma~\ref{lem:exhaustive-run}'s argument.  If
$z$ divides $k$, then exactly $k/z$ such sets are needed to
accommodate the total of $k$ writers. Otherwise, the remaining
$k \!\mod z$ writers are moved to an overflow set $R_m$.  Note that for all $R_i \in \mathcal{R}$,
$2f+1 \leq|R_i| \leq n$.  Then, we distribute the registers in each
set $R_i$ on servers in $\mathcal{S}$ so that every register in $R_i$
is mapped to a different server (i.e., $|\delta(R_i)| = |R_i|$).
Figure~\ref{fig:alg} demonstrates a possible mapping from
registers to servers.  All in all, we use
$\Sigma_{R_i \in \mathcal{R}} |R_i| = \big\lfloor \frac{k}{z}
\big\rfloor y + (k - \big\lfloor \frac{k}{z} \big\rfloor z)f +
(f+1)(\big\lceil \frac{k}{z} \big\rceil - \big\lfloor \frac{k}{z}
\big\rfloor)= \cdots = kf + \big\lfloor \frac{k}{z} \big\rfloor (f+1)
= kf + \big\lceil \frac{k}{\big\lfloor \frac{n-(f+1)}{f} \big\rfloor}
\big\rceil (f+1)$ registers.

\begin{wrapfigure}{l}{0.3\linewidth} 
   \centering 
  \includegraphics[width=0.85\linewidth]{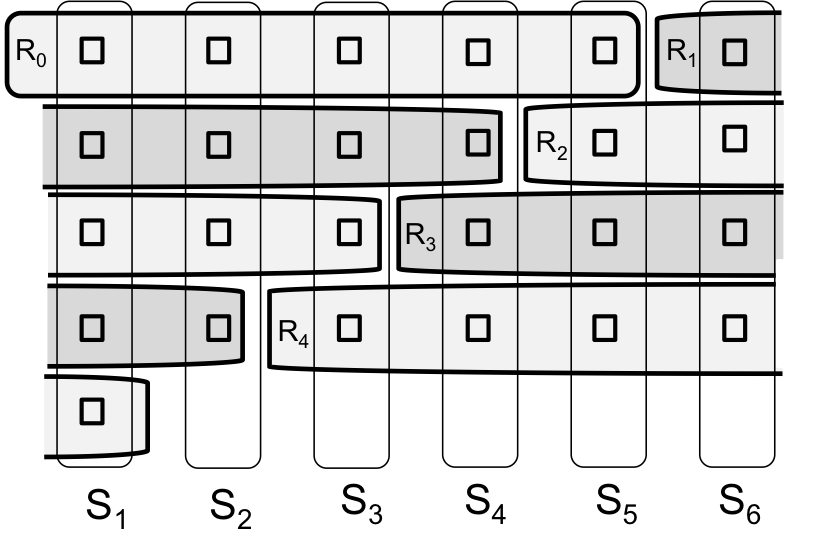}
\vspace*{-0.3cm}
  \caption {\footnotesize A possible mapping from
  $\mathcal{R}$ to $\mathcal{S}$ in case $n=6, k=5$, and $f=2$.} 
\vspace*{-0.7cm}
   \label{fig:alg}
\end{wrapfigure}
 
The resulting layout is then used to derive the {\em read}\/ and {\em
  write quorums}\/ as follows: for every set $R_i \in \mathcal{R}$,
any subset of $R_i$ of size $ |R_i| - f$ is a write quorum for all
writers $c_j$ such that $j = \big\lfloor \frac{i}{z} \big\rfloor$; and
any subset of registers consisting of {\em all}\/ registers mapped to
$ n-f$ servers is a read quorum (i.e., the set of the read quorums is
$\{B \subseteq \mathcal{B} : \exists S \in \mathcal{S}$ s.t.\
$|S| = 2f+1 \wedge \delta^{-1}(S) = B \}$).
%
%
Observe that by construction of $\mathcal{R}$, for every
set $R_i \in \mathcal{R}$, (1) the number of clients mapped to
write quorums in $R_i$ is
$\big\lfloor \frac{|R_i|-(f+1)}{f} \big\rfloor =
\frac{|R_i|-(f+1)}{f}$,
and (2) any write quorum in $R_i$ intersects with any read quorum on
at least $|R_i| - f$ registers.  Therefore, in a write-sequential run,
the latest written value is always guaranteed to be available to
subsequent \Read s provided every writer $c$ executing a
high-level \Write\ $W$ leaves no more than $f$ pending low-level
writes upon $W$'s completion.  To enforce the latter, $c$ is precluded from triggering
new low-level writes on registers on which it still has writes pending
from preceding high-level \Write s invocations.
In addition, since the registers in every (read
or write) quorum are mapped to exactly $n-f$ servers, each
quorum access is guaranteed to terminate, and thus, the
algorithm is wait-free.
In Appendix~\ref{app:upper} we give the algorithm's full details,
present its pseudo-code, and prove the following result:

\begin{theorem}
\label{thm:upper}
For all $k>0$, $f>0$, and $n > 2f$, there exists an $f$-tolerant
algorithm emulating a wait-free WS-Regular \kreg\/ using a collection
of $n$ servers storing $kf + \big\lceil \frac{k}{z}
\big\rceil (f+1)$ wait-free $z$-writer/multi-reader
atomic base registers where $z=\big\lfloor
\frac{n-(f+1)}{f} \big\rfloor$.

\noindent

%
%
%
  
\end{theorem}

\vspace{-4mm}
\section{Discussion and Future Directions}
\label{sec:conclusion}

\section{Discussion and Future Work}
\label{sec:conclusion}

We studied space complexity of emulaiting an $f$-tolerant register
from fault-prone base objects as a function of the base object type,
the number of writers $k$, the number of available servers $n$, and
the failure threshold $f$. For the three object types considered, we
established a sharp separation (by factor $k$) between registers and
both max-registers and CAS in terms of the number of objects of the
respective types required to support the emulation; we showed that no
such separation between max-registers and CAS exists.
Interestingly, these results shed light on the resource complexity
bounds in the standard shared memory model (i.e., without object
failures) as evidenced by our proof of a lower bound on the number of
registers required for implementing a \maxreg\/ for $k$ writers.
Our main technical contribution comprises the lower bound of
$ \Big\lceil \frac{k}{\frac{n-(f+1)}{f}} \Big\rceil (f+1) + kf $ and
the upper bound of
$ \Big\lceil \frac{k}{\lfloor\frac{n-(f+1)}{f} \rfloor} \Big\rceil
(f+1) + kf$
on the resource complexity of emulating an $f$-tolerant $k$-writer
register from $n$ fault-prone servers storing read/write registers.
To strengthen our lower bound, it was proved for emulations satisfying
weak liveness and safety properties.

\paragraph{Future directions.}
First, for some choices of $k$ and $n$, our bounds leave a small gap
that can be closed.  Second, an interesting question that arises is
whether our lower bound is tight for stronger properties.  In the
special case of $n=2f+1$ servers, emulation with stronger
regularity~\cite{shao2011multiwriter} is possible with $(2f+1)k$
registers (tight to our lower bound).  However, the question of the
general case ($n\geq 2f+1$) remains open.  In addition, since
atomicity usually requires readers to write, it is interesting to
investigate whether the space complexity (assuming read/write
registers) in this case also linearly depends on the number of
readers.

Our results suggest a new classification of the data types based on
space complexity of fault-tolerant emulations built from base objects
of these types, which is fundamentally different from those
established by~\cite{herlihy1991wait}
and~\cite{ellen2016complexity}. A promising future direction is to
extend this classification with additional types (e.g., multiple
assignment), and potentially generalize it into a full-fledged
hierarchy of its own.

Another possible direction is to consider the time complexity of the
emulations.  For example, we showed that although a \maxreg\/ can be
implemented from a single CAS, the time complexity of the
implementation is high. An interesting open question is to determine
whether this tradeoff is inherent.

%

\newpage

\bibliographystyle{plain}

\newpage
\appendix
%
%
%
%
%
%

\section{Formal Model}
\label{app:formal-model}

\vspace{-2mm}
\subsection{Shared Objects}

A {\em shared object}\/ supports concurrent execution of {\em
  operations}\/ performed by some set,
${\mathbb C} = \{c_1, c_2, \dots\}$, of client processes. Each operation
has an {\em invocation}\/ and {\em response}. An object {\em
  schedule}\/ is a sequence of the operation invocations and their
responses. An invoked operation is {\em complete}\/ in a given
schedule if the operation's response is also present in the schedule,
and {\em pending}\/ otherwise.  For a schedule $\sigma$, $ops(\sigma)$
denotes the set of all operations that were invoked in $\sigma$, and
$complete(\sigma)$ (resp., $pending(\sigma)$) denotes the subset of
$ops(\sigma)$ consisting of all the complete (resp., pending)
operations. Also, for a set $X$ of operations, we use $\sigma | X$ to
denote the subsequence of $\sigma$ consisting of all the invocation
and responses of the operations in $X$.

An operation $op$ {\em precedes}\/ an operation $op'$ in a schedule
$\sigma$, denoted $op \prec_\sigma op'$ iff $op'$ is invoked after
$op$ responds in $\sigma$. Operations $op$ and $op'$ are {\em
  concurrent}\/ in $\sigma$, if neither one precedes the other. A
schedule with no concurrent operations is {\em sequential}.
Given a schedule $\sigma$, we use $\sigma | i$ to denote the
subsequence of $\sigma$ consisting of the actions client $c_i$.  The
schedule is {\em well-formed}\/ if each $\sigma | i$ is a sequential
schedule.  In the following, we will only consider well-formed
schedules. 

The object's {\em sequential specification}\/ is a collection of the
object's sequential schedules in which all operations are complete.
For an object schedule $\sigma$, a {\em linearization}\/ $L_\sigma$ of
$\sigma$ is a sequential schedule consisting of all operations in
$complete(\sigma)$ along with their responses and a
subset of $pending(\sigma)$, each of which being
assigned a matching response, so that $L_\sigma$ satisfies both the
$\sigma$'s operation precedence relation ($\prec_\sigma$) and the
object sequential specification.

\vspace{-2mm}
\subsection{Registers} 

A {\em read/write register}\/ object (or simply a {\em
register}) supports two operations of the form $\lwrite(v)$,
$v\in Vals$, and $\lread()$ returning $ack$ and $v\in Vals$
respectively where $Vals$ is the register value domain. Its
sequential specification is the collection of all sequential
schedules where every $read$ returns the value written by the
last preceding $write$ or an initial value $v_0 \in Vals$ if no
such write exists.

A register is {\em multi-writer (MW)}\/ (resp., {\em
multi-reader (MR)}) if it can be written (resp., read) by an
{\em unbounded}\/ number of clients. A $k\textrm{-writer}$
register, or simply, \kreg, is a register that can be written by
at most $k>0$ distinct clients. A register is {\em single-writer
(SW)}\/ (resp., {\em single-reader (SR)}) if only one process
can write (resp., read) the register.
For a register schedule $\sigma$, we use $writes(\sigma)$ and
$reads(\sigma)$ to denote the sets of all write and read
operations invoked in $\sigma$.
We say that $\sigma$ is {\em write-sequential}\/ if no two
writes in $writes(\sigma)$ are concurrent.

\vspace{-2mm}
\subsection{Consistency Conditions} 
\label{sec:consistency}

{\em Consistency conditions}\/ specify the shared object behaviour
when accessed concurrently by the clients. Below, we
introduce a number of consistency conditions that will be used throughout this
paper.  They are expressed as a set of schedules $C$ satisfying one of
the following requirements:

\begin{description}

\item \textbf{Atomicity}~\cite{herlihy} For all schedules
  $\sigma \in C$, $\sigma$ has a linearization.

\item \textbf{Write-Sequential Regularity (WS-Reg)} For all
$\sigma \in C$, if $\sigma$ is write-sequential, then for each
$rd \in reads(\sigma) \cap complete(\sigma)$ there is a
linearization $L_{rd}$ of $\sigma | writes(\sigma) \cup \{rd\}$.

\item \textbf{Write-Sequential Safety (WS-Safe)} As WS-Reg,
but only required to hold for complete reads that are not
concurrent with any writes.

\end{description}

\vspace{-6mm}
\subsection{System Model}

We consider an {\em asynchronous fault-prone shared memory
  system}~\cite{JCT98} consisting of a set of shared {\em base}\/
objects ${\cal B} = \{b_1,b_2,\dots\}$. The objects are accessed by a
collection of clients in the set ${\mathbb C} = \{c_1, c_2, \dots\}$.

We consider a slight generalization of the model in~\cite{JCT98} where
the objects are mapped to a set ${\cal S} = \{s_1, s_2, \dots\}$ of
servers via a function $\delta$ from ${\cal B}$ to ${\cal S}$. For
$B \subseteq {\cal B}$, we will write $\delta(B)$ to denote the {\em
  image}\/ of $B$, i.e., $\delta(B) = \{\delta(b) : b \in B\}$.
Conversely, for $S \subseteq {\cal S}$, we will write $\delta^{-1}(S)$
to denote the {\em pre-image}\/ of $S$, i.e.,
$\delta^{-1}(S) = \{b : \delta(b) \in S\}$.  Note that for all
$B \subseteq {\cal B}$, $|\delta(B)| \le |B|$, and conversely, for all
$S\subseteq \mathcal{S}$, $|\delta^{-1}(S)| \ge |S|$.  Both servers
and clients can fail by crashing. A crash of a server causes all
objects mapped to that server to instantaneously crash\footnote{Note
  that the original faulty shared model of~\cite{JCT98} can be derived
  from our model by choosing $\delta$ to be an injective function.}.

 
We study algorithms emulating reliable $k\textrm{-writer}$
registers to a set of clients. Clients interact with the emulated register via {\em
  high-level}\/ read and write operations. To distinguish the
high-level emulated reads and writes from low-level base object
invocations, we refer to the former as \Read\/ and \Write. We say that
high-level operations are {\em invoked}\/ and {\em return} whereas
low-level operations are {\em triggered}\/ and {\em respond}. A
high-level operation consists of a series of trigger and respond {\em
  actions}\/ on base objects, starting with the operation's invocation
and ending with its return. Since base objects are crash-prone, we
assume that the clients can trigger several operations in a row
without waiting for the previously triggered operations to respond.

An emulation algorithm $A$ defines the behaviour of clients as
deterministic state machines where state transitions are associated
with actions, such as trigger/response of low-level operations. A {\em
  configuration}\/ is a mapping to states from system components,
i.e., clients and base objects. An {\em initial configuration}\/ is
one where all components are in their initial states.

A {\em run}\/ of $A$ is a (finite or infinite) sequence of alternating
configurations and actions, beginning with some initial configuration,
such that configuration transitions occur according to $A$. We use the
notion of time $t$ during a run $r$ to refer to the configuration
reached after the $t$\textsuperscript{th} action in $r$. A {\em run
  fragment}\/ is a contiguous sub-sequence of a run. A run is {\em
  write-only}\/ if it has no invocations of the high-level
  \Read\/ operations.

We say that a base object, client, or server is {\em faulty}\/ in a
run $r$ if it fails at some time in $r$, and correct, otherwise. A run
is {\em fair}\/ if (1) for every low-level operation triggered by a
correct client on a correct base object, there is eventually a
matching response, and (2) every correct client gets infinitely
many opportunities to both trigger a low-level operation and
execute the return actions. We say that a low-lever operation on
a base object is {\em pending}\/ in run $r$ if it was triggered but has no matching
response in $r$.
We assume that the base objects are atomic (as defined in
Section~\ref{sec:consistency}) 

\vspace{-2mm}
\subsection{Properties of the Emulation Algorithms}

\noindent %
\textbf{Safety} The emulation algorithm safety will be expressed in
terms of the consistency conditions specified in
Section~\ref{sec:consistency}. An emulation algorithm $A$ {\em
  satisfies}\/ a consistency condition $C$ if for all
  $A$'s runs $r$, the subsequence of $r$ consisting of the invocations and responses of the high-level read and write operations is a schedule in $C$.

\noindent %
{\bf Liveness} We consider the following liveness conditions that must
be satisfied in fair runs of an emulation algorithm. A {\em
  wait-free}\/ object is one that guarantees that every high-level
operation invoked by a correct client eventually returns, regardless
of the actions of other clients. 
An {\em obstruction-free}\/ object guarantees that every
high-level operation invoked by a correct client that is not
concurrent to any other operation by a correct client eventually
returns.

\noindent
{\bf Fault-Tolerance} The emulation algorithm is $f$-tolerant if it
remains correct (in the sense of its safety and liveness properties)
as long as at most $f$ servers crash for a fixed $f > 0$.

\noindent
{\bf Complexity measures} The {\em resource consumption} of an
emulation algorithm $A$ in a (finite) run $r$ is the number of
base objects used by $A$ in $r$. The {\em resource
complexity}~\cite{JCT98} of $A$ is the maximum resource
consumption of $A$ in all its runs. 

%
%

\section{A \maxreg\ Emulation With One CAS}
\label{app:CAStoMax}

We present here a wait-free emulation of an atomic max-register
on top of a single CAS object.
The pseudocode appears in Algorithm~\ref{alg:CAStoMAX}.
The CAS object supports one operation with two parameters, $exp$
and $new$; if $exp$ is equal to the object's current value, then
the value is set to $new$.
In any case, the operation returns the old value.

A max-register supports two operations, $\emph{\writemax}(v)$
for some $v$ from some domain of ordered values $\mathbb{V}$
that returns ok, and $\emph{\readmax}()$ that returns a value
from $\mathbb{V}$.
The sequential specification of max-register is the following:
A $\emph{\readmax}$ returns the highest value among those
written by $\emph{\writemax}$ before it, or $v_0$ in case
no such values.

\begin{algorithm}[H]
\label{alg:CAStoMAX}

\caption{\maxreg\ emulation from a single CAS object $b$}

\begin{algorithmic}[0]
\small
\Local{}{}
 \State $tmp \in \mathbb{V}$, initially $v_0$
 \EndLocal

\Operation{$b.CAS(exp, new)$, $exp, new \in \mathbb{V}$}{}
	\State $prev \gets b$
	\If{$exp = b$}
		\State $b \gets new$
	\EndIf
	\State {\bf return} $prev$ \label{line:cas-end}
\EndOperation

\end{algorithmic}

\label{alg:CAStoMAX}
\begin{algorithmic}[1]
\small%

\Operation{$\emph{\writemax}(v)$}{} %

	\While{true}
		\State $tmp \gets b.CAS(v_0, v_0)$
		\label{line:writeLinear2}
		\If{$tmp \geq v$}
			\State return ok
		\EndIf
		\State $b.CAS(tmp, v)$
		\label{line:writeLinear1}
	
	\EndWhile
\EndOperation %

\Statex

\Operation{\emph{\readmax}()}{} %

	\State  $tmp \gets b.CAS(v_0, v_0)$
	\label{line:readLinear}
	\State return $tmp$
	
\EndOperation %

\end{algorithmic}
\end{algorithm}

\subsection{Correctness}

We say that a $b.CAS(exp,new)$ operation is \emph{successful} if
$b$ is set to $new$.

\begin{observation}
\label{obs:successful}

Consider a successful $b.CAS(exp,new)$ operation for some $exp$
and $new$, the the next $b.CAS(exp',new')$ operation for some
$exp'$ and $new'$ returns $new$.

\end{observation}

\noindent The following observation follows immediately
from the fact that $b.CAS(exp, new)$ is called only with $new
\geq exp$.

\begin{observation}
\label{obs:monotonic}

The values returned by $b.CAS(exp, new)$ are monotonically
increasing.

\end{observation}

%
%
%

\noindent We next define linearization points:

\begin{definition} (linearization points)

\begin{description}
	\item[\emph{\readmax}:] The linearization point is
	line~\ref{line:readLinear}.

	\item[\emph{\writemax}:] If the operation performs a successful
	$b.CAS(tmp, v)$ in line~\ref{line:writeLinear1}, then the
	linearization point is the last time
	Line~\ref{line:writeLinear1} is performed.
	Otherwise, the linearization point is last time
	line~\ref{line:writeLinear2} is performed.

\end{description}

\end{definition}

\begin{lemma}
\label{lem:CASlinearization}

For every run $r$ of Algorithm~\ref{alg:CAStoMAX}, the sequential
run $\sigma_r$, produced by the linearization points of
operations in $r$, is a linearization of $r$.

\end{lemma}

\begin{proof}

The real time order of $r$ is trivially preserved in $\sigma_r$.
We need to show that $\sigma_r$ preserves max-register's
sequential specification.
Let $mr$ be a \emph{\readmax()} in $r$ that returns a value $v$,
and let $t_{mr}$ be the time when $b.CAS(v_0,v_0)$ (that
returns $v$) is called by $mr$ (line~\ref{line:readLinear}).
We need to show that (1) there is no \emph{\writemax$(v')$} that
precedes $mr$ in $\sigma_r$ s.t.\ $v'>v$, and (2) if $v \neq
v_0$, there is a \emph{\writemax$(v)$} that precedes $mr$ in
$\sigma_r$


\begin{enumerate}
  
  \item Assume by way of contradiction that there is a
  \emph{\writemax$(v')$} $w$ that precedes $mr$ in $\sigma_r$
  s.t.\ $v'>v$. Denote $t_w$ to be the linearization point of
  $w$ in $r$, and note that $t_w < t_{mr}$. 
  Now consider two case:
  
  \begin{enumerate}
    
    \item $w$ performs line~\ref{line:writeLinear1} at time
    $t_w$ (line~\ref{line:writeLinear1} is $w$'s linearization
    point).
    In this case, $w$ performs a successful $b.CAS(tmp, v')$ in
    line~\ref{line:writeLinear1} for some $tmp$ at some time $t'
    \leq t_w < t_{mr}$. Therefore, by
    Observations~\ref{obs:successful} and~\ref{obs:monotonic},
    $b.CAS(v_0, v_0$) called by $mr$ in
    line~\ref{line:readLinear} at time $t_{mr}$ returns $v'' \geq v' > v$. Hence, $mr$ returns $v'' \geq v' > v$.
    A contradiction.
    
    \item $w$ performs line~\ref{line:writeLinear2} at time
    $t_w$ (line~\ref{line:writeLinear2} is $w$'s linearization
    point).
    Thus, $t_w$ is the last time $w$ calls $tmp \gets b.CAS(v_0,
    v_0)$ in lime~\ref{line:writeLinear2}.
    Therefore, $b.CAS(v_0, v_0)$, called at time $t_w$ returns
    a value bigger than or equal to $v'$.
    Therefore, by
    Observations~\ref{obs:successful} and~\ref{obs:monotonic},
    $b.CAS(v_0, v_0)$ called by $mr$ in
    line~\ref{line:readLinear} at time $t_{mr}$ returns $v''
    \geq v' > v$. Hence, $mr$ returns $v'' \geq v' > v$.
    A contradiction. 
    
  \end{enumerate}

  \item Assume that $v \neq v_0$.
  By the code and by the CAS properties, there is a
  \emph{\writemax(v)} operation $w$ that calls a successful
  $b.CAS(tmp, v)$ (line~\ref{line:writeLinear1}) for some $tmp$
  at time $t_w < t_{mr}$.
  Therefore, by Observations~\ref{obs:successful}
  and~\ref{obs:monotonic}, the next call to $b.CAS(v_0, v_0)$
  (in line~\ref{line:writeLinear2}) by $w$ returns $tmp \geq v$,
  and thus $w$ does not perform line~\ref{line:writeLinear1}
  again.
  We get that $t_w$ is the linearization point of $w$, and thus
  $w$ precedes $mr$ in $\sigma_r$.

\end{enumerate}

\end{proof}


%
%
%

\begin{theorem}

Algorithm~\ref{alg:CAStoMAX} emulates wait-free atomic
max-register.

\end{theorem}

\begin{proof}

Atomicity follows from Lemma~\ref{lem:CASlinearization}.
We left to show wait-freedom:

\begin{itemize}
  
  \item \textbf{\readmax():} Since $b.CAS(exp,new)$ is wait free,
  \readmax() is wait-free.
  
  \item \textbf{\writemax(v):} Note that $\emph{\writemax}(v)$
  returns in the iteration in which $b.CAS(v_0, v_0)$ in
  line~\ref{line:writeLinear2} returns a value bigger than or
  equal to $v$.
  Therefore, by Observations~\ref{obs:successful}
  and~\ref{obs:monotonic}, $\emph{\writemax}(v)$ returns in the
  following iteration after a successful $b.CAS(tmp, v)$ in
  lime~\ref{line:writeLinear1}.
  Now assume in a way of contradiction that 
  $\emph{\writemax}(v)$ do not have a successful
  $b.CAS(tmp,v)$ in lime~\ref{line:writeLinear1}.
  By the code and Observations~\ref{obs:successful}
  and~\ref{obs:monotonic}, if $b.CAS(tmp, v)$ in
  lime~\ref{line:writeLinear1} do not succeed, then the
  following $b.CAS(v_0, v_0)$ in line~\ref{line:writeLinear2}
  returns a bigger value that what it returned in the previous
  iteration.
  Now let $v'$ be the value returned by the
  first $b.CAS(v_0, v_0)$ in line~\ref{line:writeLinear2}, and
  assume w.l.o.g. that there are $k$ values bigger than $v'$ and
  smaller than $v$ in $\mathbb{V}$.
  Therefore, after at most $k$ iteration $\emph{\writemax}(v)$
  returns.

\end{itemize}

\end{proof}
\newpage


\section{Space Lower Bounds}
\label{sec:space:app}

\begin{lemclone}{lem:facts1}~
  Let $r\in \tup{r_{i-1}}{Ad_i}$ be a run consisting of $t$ steps.
  Then, for all $t \ge t_{i-1}$, all of the following hold:
\begin{AutoMultiColItemize}
\item $Q_i(t) \subseteq \delta(Cov_i(t)) \setminus F$ \label{lem:facts:qsubset}
\item $Q_i(t) \subseteq Q_i(t+1)$  \label{obs:ad}
\item $F_i(t) \subseteq F_i(t+1)$  \label{obs:adf}
\item $|F_i(t)| - |Q_i(t)| \leq 1$ \label{obs:ad1}
\item $|Q_i(t)| \le f$ \label{obs:qlim}
\item $|F_i(t)| \le f+1$ \label{obs:filim}
\item $F_i(t) = F_i(t+1)\ \implies M_i(t) \subseteq M_i(t+1)$
 \label{lem:facts:covfi}
\item $|M_i(t)| \le f+1$ \label{lem:facts:maxcovfi}
\item $|\delta(Cov_i(t)) \setminus F| \ge f \implies |Q_i(t)|\ge f$
\label{lem:facts:qisize}
\item $|\delta(Cov_i(t)) \setminus F| < f \implies \delta(Rr_i(t)) \setminus F = \emptyset$
\label{lem:facts:rrr}
\item $(Q_i(t) \cup M_i(t)) \cap \delta^{-1}(Rr_i(t)) = \emptyset$
\label{lem:facts:silent}
\end{AutoMultiColItemize}
\end{lemclone}

\begin{proof}
  By induction on $t\ge t_{i-1}$.\\
  {\em Base}: If $t = t_{i-1}$, then
  $Tr_i(t) = Rr_i(t) = Cov_i(t) = F_i(t) = \emptyset$. Furthermore,
  since $|\delta(Cov_i(t)) \setminus F| = 0 \le 1 \le f$,
  $Q_i(t) = \delta(Cov_i(t)) \setminus F$. Thus, all the claims hold.
  
  \noindent
  {\em Induction step}: Suppose all the claims hold for all
  $t\ge t_{i-1}$, and consider the time $t+1$:

  \ref{lem:facts1}.\ref{lem:facts:qsubset}: If
  $|\delta(Cov_i(t+1)) \setminus F| \le f$, then by
  Definition~\ref{def:notation}.\ref{def:notation:qi},
  $Q_i(t+1) = \delta(Cov_i(t+1)) \setminus F$ as needed. Otherwise,
  $Q_i(t+1) = Q_i(t)$. Consider an arbitrary server $s\in Q_i(t+1)$,
  and towards a contradiction, suppose that
  $s \notin \delta(Cov_i(t+1)) \setminus F$. Since
  $s\in Q_i(t+1) = Q_i(t)$, by the induction hypothesis,
  $s \in \delta(Cov_i(t)) \wedge s \notin F$. Since
  $s\notin \delta(Cov_i(t+1)) \setminus F$, we get that either (1)
  $s\notin \delta(Cov_i(t+1))$ or (2)
  $s \in \delta(Cov_i(t+1)) \cap F$. Since $s \notin F$, (2) is false,
  and therefore, $s\notin \delta(Cov_i(t+1))$. Hence, there exists a
  base register in $\delta^{-1}(\{s\})$ that responded at $t$ to a
  low-level write $w$ triggered after $t_{i-1}$. Since $s\in Q_i(t)$,
  by Definition~\ref{def:blocked-writes}, $w\in \BlockedWrites(t)$.
  However, since the environment behaves like $Ad_i$ after $r_{i-1}$,
  by Definition~\ref{def:ad}.\ref{def:ad:block}, $w$ does not respond
  at $t$. A contradiction.

  \ref{lem:facts1}.\ref{obs:ad}: Towards a contradiction, suppose that
  there exists $s \in {\cal S}$ such that
  $s \in Q_i(t) \wedge s \not\in Q_i(t+1)$. By
  Definition~\ref{def:notation}.\ref{def:notation:qi},
  $|\delta(Cov_i(t+1)) \setminus F| \le f$ as otherwise,
  $Q_i(t+1) = Q_i(t)$ contradicting the assumption. Thus,
  $Q_i(t+1) = \delta(Cov_i(t+1)) \setminus F$, and therefore, either
  (1) $s \not\in \delta(Cov_i(t+1))$ or (2)
  $s \in \delta(Cov_i(t+1)) \cap F$.  By the induction hypothesis for
  \ref{lem:facts1}.\ref{lem:facts:qsubset},
  $s \in \delta(Cov_i(t)) \wedge s \not\in F$. Hence, (2) is false,
  and it is only left to consider the case
  $s \not\in \delta(Cov_i(t+1))$. Thus, $s \in \delta(Cov_i(t))$ and
  $s \not\in \delta(Cov_i(t+1))$, which implies that there exists a
  base register in $\delta^{-1}(\{s\})$ that responded at time $t$ to
  a low-level write $w$ invoked after $t_{i-1}$. Since $s\in Q_i(t)$,
  by Definition~\ref{def:blocked-writes}, $w\in \BlockedWrites(t)$.
  However, since the environment behaves like $Ad_i$ after $r_{i-1}$,
  by Definition~\ref{def:ad}.\ref{def:ad:block}, $w$ does not respond
  at $t$. A contradiction.
  
  \ref{lem:facts1}.\ref{obs:adf}: Let $s\in F_i(t)$. By
  Definition~\ref{def:notation}.\ref{def:notation:fi}, there exists a
  base register $b \in \delta^{-1}(\{s\})$ that responded to a low-level
  write triggered on $b$ after $t_{i-1}$ at time $t'$ such that
  $t_{i-1} < t' \le t < t+1$. Since $t < t+1$, the $b$'s response has
  also occurred before $t+1$, and therefore, $b\in Rr_i(t+1)$. Hence,
  $b \in (\delta^{-1}(S) \cap Rr_i(t+1)) = F_i(t+1)$ as needed.

  \ref{lem:facts1}.\ref{obs:ad1}: Toward a contradiction, suppose that
  $|F_i(t+1)| - |Q_i(t+1)| > 1$. Since we already proved that
  $Q_i(t) \subseteq Q_i(t+1)$, $|Q_i(t)| \le |Q_i(t+1)|$. In addition,
  we know that $||Q_i(t+1)| - |Q_i(t)|| \le 1$,
  $||F_i(t+1)| - |F_i(t)|| \le 1$, and by the induction hypothesis
  $|F_i(t)| - |Q_i(t)| \le 1$. Thus, $|F_i(t+1)| - |Q_i(t+1)| > 1$
  implies that (1) $|F_i(t)| - |Q_i(t)| = 1$ (i.e.,
  $|F_i(t)| > |Q_i(t)|$), (2) $|Q_i(t+1)| = |Q_i(t)|$, and (3)
  $|F_i(t+1)| = |F_i(t)| + 1$.  Since we already proved that
  $F_i(t+1) \supseteq F_i(t)$, (3) implies that there exists
  $s \in {\cal S}$ such that $s \in F_i(t+1) \setminus F_i(t)$. Since
  by Definition~\ref{def:notation}.\ref{def:notation:fi},
  $F_i(t+1) \subseteq F$, $s \in F$. Thus, $s \in F \setminus F_i(t)$.
  This means that either no low-level writes have been triggered on
  registers in $\delta^{-1}(\{s\})$ after $t_{i-1}$, or there is a
  register $b\in \delta^{-1}(\{s\})$ that responds to a low-level
  write triggered on $b$ after $t_{i-1}$. In the first case, no
  register on $s$ can respond at time $t$, and therefore,
  $s \not\in F_i(t+1)$, which is a contradiction. In the second case,
  we obtain that $b$ satisfies $b\in Cov_i(t)$,
  $b\in \delta^{-1}(\{s\})$, $s \in F \setminus F_i(t)$, and $b$
  responds at $t$ to a covering write $w$ triggered after $t_{i-1}$.
  Thus, by Definition~\ref{def:notation}.\ref{def:notation:mi},
  $s \in M_i(t)$ and since $|F_i(t)| > |Q_i(t)|$, by
  Definition~\ref{def:notation}.\ref{def:notation:gi}, $s\in G_i(t)$.
  Thus, by Definition~\ref{def:blocked-writes},
  $w\in \BlockedWrites(t)$, and since the environment behaves like
  $Ad_i$ after $r_{i-1}$, by
  Definition~\ref{def:ad}.\ref{def:ad:block}, $w$ does not respond at
  $t$. A contradiction.

  \ref{lem:facts1}.\ref{obs:qlim}: Assume by contradiction that
  $|Q_i(t+1)| > f$.  Since by \ref{lem:facts1}.\ref{lem:facts:qsubset},
  $|Q_i(t+1)| \subseteq \delta(Cov_i(t+1)) \setminus F$,
  $|\delta(Cov_i(t+1)) \setminus F| > f$. By
  Definition~\ref{def:notation}.\ref{def:notation:qi},
  $Q_i(t+1)=Q_i(t)$, and therefore, $|Q_i(t)| > f$. A contradiction to
  the inductive assumption.

  \ref{lem:facts1}.\ref{obs:filim}: By
  Definition~\ref{def:notation}.\ref{def:notation:fi},
  $F_i(t+1) \subseteq F$. Since $|F|=f+1$, $|F_i(t+1)| \le f+1$.

  \ref{lem:facts1}.\ref{lem:facts:covfi}: Suppose $F_i(t) = F_i(t+1)$.
  Consider $s \in M_i(t)$, and toward a contradiction, suppose that
  $s\not\in M_i(t+1)$. Since by
  Definition~\ref{def:notation}.\ref{def:notation:fi},
  $F_i(t) \subseteq F$ and $F_i(t+1) \subseteq F$,
  $F\setminus F_i(t) = F \setminus F_i(t+1)$. This together with the
  fact that $s\in F\setminus F_i(t)$ implies that
  $s\in F\setminus F_i(t+1)$ as well. Thus, it must be the case that
  $s \in \delta(Cov_i(t)) \wedge s \not\in \delta(Cov_i(t+1))$. Thus,
  by Definition~\ref{def:notation}.\ref{def:notation:covi}, there
  exists a base register $b\in \delta^{-1}(\{s\})$ that responds to a
  low-level write invoked after $t_{i-1}$ at time $t$ which implies
  that $b\in Rr_i(t+1)$. Hence, by Definition~\ref{def:notation:fi},
  $s\in F_i(t+1)$. However, since $s\in F\setminus F_i(t+1)$,
  $s\not\in F_i(t+1)$. A contradiction.
  
  \ref{lem:facts1}.\ref{lem:facts:maxcovfi}: Since $F$ is fixed in advance
  and $|F|=f+1$, we receive
  $|M_i(t+1)|=|\delta(Cov_i(t+1)) \cap (F\setminus F_i(t+1))| \le |F\setminus
  F_i(t+1)| \le |F| = f+1$.

  \ref{lem:facts1}.\ref{lem:facts:qisize}: If
  $|\delta(Cov_i(t)) \setminus F| < f$ and
  $|\delta(Cov_i(t+1)) \setminus F| \ge f$, then there exists a base
  register on a server in $\mathcal{S} \setminus F$ that is newly
  covered after $t$. Thus, we get
  $|\delta(Cov_i(t)) \setminus F| = f - 1$ and
  $|\delta(Cov_i(t+1)) \setminus F| = f$. By
  Definition~\ref{def:notation}.\ref{def:notation:qi},
  $Q_i(t+1) = \delta(Cov_i(t+1)) \setminus F$, and therefore,
  $|Q_i(t+1)|=f$. Otherwise, by the induction hypothesis,
  $|Q_i(t)|\ge f$.  Since $|\delta(Cov_i(t+1)) \setminus F| \ge f$, we
  have that either (1) $|\delta(Cov_i(t+1)) \setminus F| = f$, or (2)
  $|\delta(Cov_i(t+1)) \setminus F| > f$. Applying
  Definition~\ref{def:notation}.\ref{def:notation:qi} to (1) and (2),
  we get the following: for (1),
  $Q_i(t+1) = \delta(Cov_i(t+1))\setminus F$, which implies
  $|Q_i(t+1)| = |\delta(Cov_i(t+1))\setminus F| = f$, and for (2),
  $Q_i(t+1) = Q_i(t)$, and therefore, $|Q_i(t+1)|\ge f$.

  \ref{lem:facts1}.\ref{lem:facts:rrr}: Toward a contradiction,
  suppose that $|\delta(Cov_i(t+1)) \setminus F| < f$ and
  $\delta(Rr_i(t+1)) \setminus F \neq \emptyset$.  By the induction
  hypothesis,
  $|\delta(Cov_i(t)) \setminus F| < f \implies \delta(Rr_i(t))
  \setminus F = \emptyset$.
  We first consider the case $|\delta(Cov_i(t)) \setminus F| \ge f$.
  Thus given $|\delta(Cov_i(t+1)) \setminus F| < f$, there exists a
  server in $\delta(Cov_i(t)) \setminus F$ such that some register on
  that server responds to a low-level write $w$ that was triggered
  after $t_{i-1}$. Moreover, $|\delta(Cov_i(t)) \setminus F| = f$, and
  thus, by Definition~\ref{def:notation}.\ref{def:notation:qi},
  $Q_i(t) = \delta(Cov_i(t)) \setminus F$. Since the environment
  behaves like $Ad_i$ after $r_{i-1}$, by
  Definition~\ref{def:blocked-writes}, $w\in \BlockedWrites(t)$, and
  therefore, by Definition~\ref{def:ad}.\ref{def:ad:block}, $w$ does
  not respond at $t$. A contradiction.

  Thus, we know that $|\delta(Cov_i(t)) \setminus F| < f$ and
  $\delta(Rr_i(t)) \setminus F = \emptyset$. And since
  $\delta(Rr_i(t+1)) \setminus F \neq \emptyset$, there exists a
  server $s\in \delta(Cov_i(t))\setminus F$ such that some object on
  $s$ responded at $t$ to a low-level write $w$ triggered after
  $t_{i-1}$. Since $|\delta(Cov_i(t)) \setminus F| < f$, by
  Definition~\ref{def:notation}.\ref{def:notation:qi},
  $Q_i(t) = \delta(Cov_i(t)) \setminus F$. Thus, $s \in Q_i(t)$.
  However, by Definition~\ref{def:blocked-writes},
  $w\in \BlockedWrites(t)$, and since the environment behaves like
  $Ad_i$ after $t_{i-1}$, by
  Definition~\ref{def:ad}.\ref{def:ad:block}, $w$ does not respond at
  $t$. A contradiction.

  \ref{lem:facts1}.\ref{lem:facts:silent}: Toward a contradiction,
  suppose that
  $(Q_i(t+1) \cup M_i(t+1)) \cap \delta(Rr_i(t+1)) \neq \emptyset$.  We will
  consider the following two cases separately: (1)
  $Q_i(t+1) \cap \delta(Rr_i(t+1)) \neq \emptyset$, and (2)
  $M_i(t+1) \cap \delta(Rr_i(t+1)) \neq \emptyset$.

  (1) Suppose $Q_i(t+1) \cap \delta(Rr_i(t+1)) \neq \emptyset$, and
  let $s\in Q_i(t+1) \cap \delta(Rr_i(t+1))$. If $s \in Q_i(t)$, then
  by the induction hypothesis $s \not\in \delta(Rr_i(t))$. This means
  that either (a) $\delta^{-1}(\{s\}) \cap Tr(t) = \emptyset$, or (b) there
  exists a base register in $\delta^{-1}(\{s\})$ that responds to a low-level
  write $w$ triggered after $t_{i-1}$. If (a) holds, then no base
  register can respond to a low-level write before $t+1$, and
  therefore, $s\notin \delta(Rr_i(t+1))$, which is a contradiction. If
  (b) is the case, then since $s \in Q_i(t)$, by
  Definition~\ref{def:blocked-writes}, $w\in \BlockedWrites(t)$. Since
  the environment behaves like $Ad_i$ after $t_{i-1}$, by
  Definition~\ref{def:ad}.\ref{def:ad:block}, $w$ does not respond at
  $t$, which is also a contradiction. 

  If $s\notin Q_i(t)$, then since $s\in Q_i(t+1)$, by
  Definition~\ref{def:notation}.\ref{def:notation:qi}, the only action
  that can follow $t$ is a trigger of a low-level write on some
  register in $\delta^{-1}(\{s\})$. Since $s\in \delta(Rr_i(t+1))$,
  and the action executed at $t$ is not a respond,
  $s \in \delta(Rr_i(t))$. Thus, $Q_i(t+1) \neq Q_i(t)$ which by
  Definition~\ref{def:notation}.\ref{def:notation:qi}, implies that
  $Q_i(t+1) = \delta(Cov_i(t+1)) \setminus F$, and
  $|Q_i(t+1)| = |\delta(Cov_i(t+1)) \setminus F| \le f$. Since
  $Q_i(t) \subset Q_i(t+1)$, $|Q_i(t)| < f$, and by the induction
  hypothesis for \ref{lem:facts1}.\ref{lem:facts:qisize}, we have
  $|\delta(Cov_i(t+1)) \setminus F| < f$. Thus, by the induction
  hypothesis for \ref{lem:facts1}.\ref{lem:facts:rrr}, we conclude
  that no registers on servers in $\mathcal{S} \setminus F$ have
  responded to any low-level writes triggered between $t_{i-1}$ and
  $t$. However, since $s \in \delta(Rr_i(t))$ and, by the induction
  hypothesis for \ref{lem:facts1}.\ref{lem:facts:qsubset},
  $s \notin \delta(Rr_i(t))$, $s\in \mathcal{S} \setminus F$ has a
  register that responded to a low-level write triggered after
  $t_{i-1}$. A contradiction.

  (2) Suppose that $M_i(t+1) \cap \delta(Rr_i(t+1)) \neq \emptyset$,
  and let $s\in M_i(t+1) \cap \delta(Rr_i(t+1))$. By
  Definition~\ref{def:notation}.\ref{def:notation:fi}, we know that
  $s\in F\setminus F_i(t+1)$, and therefore, $s\in F$ and
  $s\notin F_i(t+1)$. Thus, $s \in F \cap \delta(Rr_i(t+1))$, and
  therefore, by Definition~\ref{def:notation}.\ref{def:notation:fi},
  $s \in F_i(t+1)$. A contradiction.
\end{proof}

We now give the full proof of Lemma~\ref{lem:2f}. 
An illustration of the runs constructed in the proof appear in
Figure~\ref{fig:runs}.\\

\begin{figure}[ht]      
  \centering 
  \includegraphics[width=6.5in]{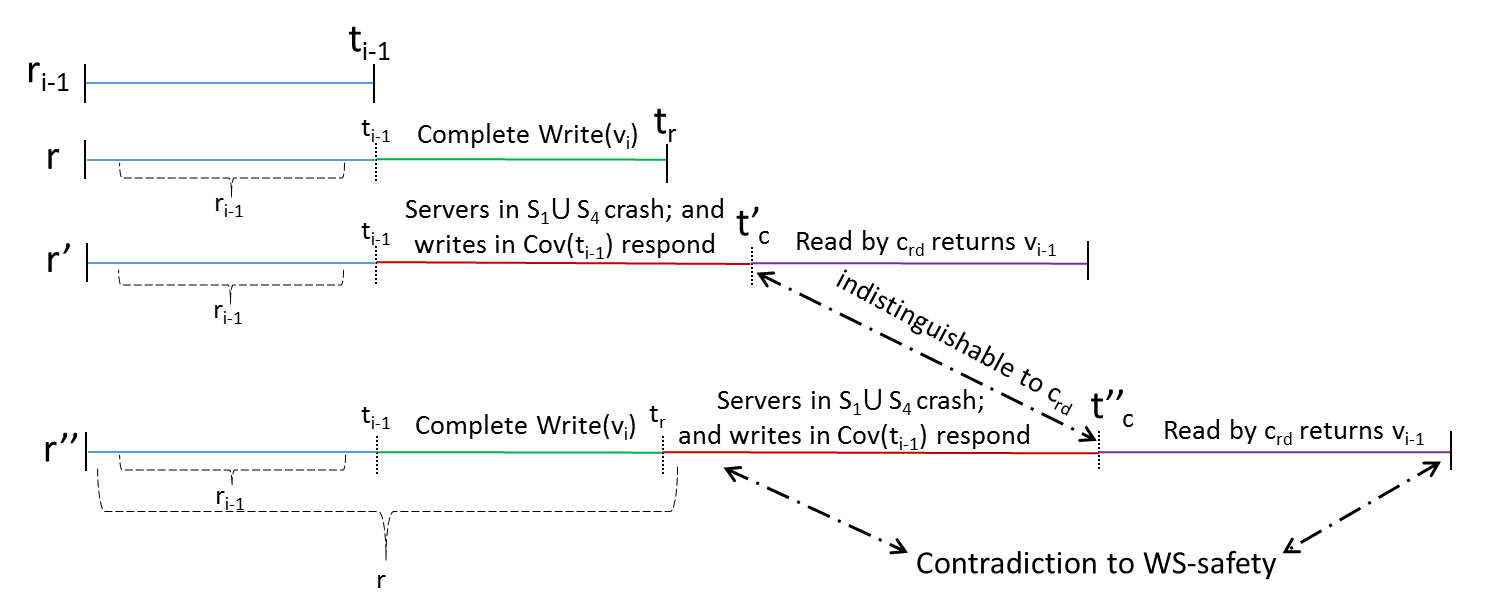}
  \caption{An illustration of the runs constructed in
  Lemma~\ref{lem:2f}.}
   \label{fig:runs}
\end{figure}

\begin{lemclone}{lem:2f}~
  Consider a run $r \in \triple{r_{i-1}}{Ad_i}{t_r}$ where
  $t_r > t_{i-1}$, consisting of $r_{i-1}$ followed by a complete
  high-level write invocation by client $c_i \not\in C(t_{i-1})$ that
  returns at time $t_r$. Then,
  $|\delta(Tr_i(t_r) \setminus Cov(t_{i-1}))| > 2f$.

\end{lemclone}

\begin{proof}
  Denote $X \triangleq \delta(Tr_i(t_r) \setminus Cov(t_{i-1}))$, and
  assume by contradiction that $|X| \leq 2f$.  Let
  $S_1= F_i(t_r)$, $S_2=Q_i(t_r)$,
  $S_3 = X \cap (F \setminus F_i(t_r))$ and
  $S_4=X \setminus (S_1\cup S_2 \cup S_3)$.  Note that
  $S_1,S_2,S_3,S_4$ are pairwise disjoint, and
  $X = S_1 \cup S_2 \cup S_3 \cup S_4$.  

  We first show that $|S_1 + S_4| \le f$. By
  Lemma~\ref{lem:facts1}.\ref{obs:filim}, $|S_1|\le f+1$.  However, if
  $|S_1|=f+1$, then by Lemma~\ref{lem:facts1}.\ref{obs:ad1},
  $|S_1| - |S_2| = f + 1 - |S_2| \le 1$, and therefore,
  $|S_1| + |S_2|\ge 2f+1$ violating the assumption. Hence,
  $|S_1| \le f$. By Lemma~\ref{lem:facts1}.\ref{obs:qlim},
  $|S_2|\le f$. If $|S_2| = f$, then by assumption,
  $|S_1 \cup S_3 \cup S_4| = |S_1|+|S_3|+|S_4| \leq f$, and
  therefore, $|S_1+S_4| = |S_1| + |S_4| \le f$. And if $|S_2|<
  f$, then by Definitions~\ref{def:notation}.\ref{def:notation:qi} and
  \ref{def:ad}, 
  $|S_4|=0$.
  Hence, $|S_1 + S_4| = |S_1| + |S_4| \le f$.
  
  Now let $r'$ be a fair extension of $r_{i-1}$ consisting of $t'_c$
  steps in which $r_{i-1}$ is followed by (1) the crash events of all
  servers in $S_1 \cup S_4$, and (2) the respond steps of all the
  covering writes in $r_{i-1}$ and (and no other steps). Extend $r'$
  with an invocation of a high-level read operation $R$ by client
  $c_{rd}\neq c_i$ at time $t'_c$.  Since $|S_1 + S_4| \le f$, by
  obstruction freedom and $f$-tolerance, there exists time
  $t_{rd} > t_{i-1}$ at which $R$ returns.  Since $r'$ is
  write-sequential, by WS-Safety, $R$ must return $v_{i-1}$.

  Next, let $r''$ be an extension of $r$ consisting of all steps in
  $r$ up to the time $t_r$ followed by (1) the crash events of all
  servers in the set $S_1 \cup S_4$, and (2) the respond steps of all
  covering writes in $r_{i-1}$ (and no other steps). Let $t''_c > t_r$
  be the number of steps in $r''$. 
  By Assumption~\ref{ass:linear}, the values that can be
  read from the base registers in $Cov(t_{i-1})$ at time
  $t''_c$ in $r''$ are identical to those that can be read
  at time $t'_c$ in $r'$.
  Furthermore, by
  definitions~\ref{def:notation}.\ref{def:notation:fi} 
  and~\ref{def:ad}, low-level writes triggered on registers in
  $\delta^{-1}(S_2 \cup S_3)$ do not respond before $t_r$ in
  $r$.
  Thus, by Assumption~\ref{ass:linear}, the values that can be
  read from the base registers in $\delta^{-1}(S_2 \cup S_3)$ at
  time $t'_c$ in $r'$ are also the same as those that can be read
  at time $t''_c$ in $r''$.  
  Thus, all registers in non-faulty
  servers at time $t'_c$ in $r'$ will appear to the subsequent
  reads as having the same content as at the time $t''_c$ in
  $r''_c$.
  
  We now extend $r''$ by letting client $c_{rd}$ invoke high-level
  read $R$ at time $t''_c$.  Since $r'$ is indistinguishable from
  $r''$ to $c_{rd}$, and $R$ has no concurrent high-level operations,
  we get that $R$ returns $v_{i-1}$ in $r''$.  However, since $W_i$ is
  the last complete write preceding $R$ in $r''$, by WS-Safety, the
  $R$'s return value must be $v_i \neq v_{i-1}$. A contradiction.
\end{proof}

\begin{thmclone}{thm:maxreg}~[Resource Complexity of $k$-\maxreg]
  Any algorithm implementing a wait-free $k$-writer \maxreg\/ from a
  collection of wait-free MWMR atomic base registers uses at least
  $k>0$ base registers.
\end{thmclone}

\begin{proof}
  Suppose to the contrary that there exists an algorithm $A$
  implementing a $k$-writer \maxreg\/ using $\ell<k$ base MWMR
  wait-free atomic registers. Consider a fault-prone shared memory
  system consisting of $n=2f+1$ servers each of which stores $\ell$
  MWMR wait-free atomic registers. Run $n$ copies $A_1,\dots,A_n$ of
  $A$, one on each server, to obtain $n=2f+1$ copies of $k$-writer
  \maxreg. Run a generic protocol of~\cite{mwr-journal} to obtain an
  $f$-tolerant emulation $\mathcal{A}$ of a wait-free $k$-writer
  regular register. By assumption, the resource complexity of
  $\mathcal{A}$ is $(2f+1)\ell < (2f+1)k$ base registers. However, by
  Theorem~\ref{thm:kf}, for $n = 2f+1$, it must be at least
  $kf + k(f+1) = (2f+1)k$. A contradiction.
\end{proof}

We next prove an extended version of Lemma~\ref{lem:exhaustive-run}
and use it to show additional lower bounds (see
Theorems~\ref{thm:2f+1}, \ref{thm:2fplus1}, \ref{thm:bounded-servers},
and \ref{theorem:adaptive}).


\begin{lemexe}{lem:exhaustive-run}~
  For all $k>0$, $f>0$, let $A$ be an $f$-tolerant algorithm that
  emulates a WS-Safe obstruction-free \kreg\/ using a collection
  $\mathcal{S}$ of servers storing a collection $\mathcal{B}$ of
  wait-free MWMR atomic registers.  Then, for every $F \subseteq S$
  such that $|F| = f +1$, there exist $k$ failure-free runs $r_i$,
  $0 \le i \le k$, of $A$ such that (1) $r_0$ is a run consisting of
  an initial configuration and $t_0=0$ steps, and (2) for all
  $i \in [k]$, $r_i$ is a write-only
  sequential extension of $r_{i-1}$
  ending at time $t_i
  > 0$ that consists of $i$ complete high-level writes of
  $i$
  distinct values $v_1,\dots,v_i$
  by $i$ distinct clients $c_1,\dots,c_i$ such that:
  
  \begin{multicols}{2}
  \begin{enumerate}[label=\alph*)]
    
  \item $|Cov(t_i)| \ge if$
    \label{lem:CovF}
    
  \item $\delta(Cov(t_i)) \cap F = \emptyset$
    \label{lem:F}
    
  \item $|\delta(Tr(t_i) \setminus Cov( t_{i-1} ))| > 2f$
    \label{lem:Tr}
    
  \item $|\delta(~Cov(t_i) \setminus Cov(t_{i-1})~)| \ge f$
    \label{lem:Cov}
   
  \item  $Cov(t_i) \supseteq Cov(t_{i-1})$
    \label{lem:CovSup}
     
  \end{enumerate}
  \end{multicols}
  
\end{lemexe}

\begin{proof}[Proof of \ref{lem:Tr}~--~\ref{lem:CovSup}]: Fix
  arbitrary $k>0$, $f>0$, and a set $F$ of servers such that
  $|F|=f+1$. We proceed by induction on $i$, $0 \le i \le k$.
  \textbf{Base:} Trivially holds for the run $r_0$ of $A$ consisting
  of $t_0=0$ steps. \textbf{Step:} Assume that $r_{i-1}$ exists for
  all $i\in [k-1]$. We use the extension $r'$ constructed in the proof
  of Lemma~\ref{lem:exhaustive-run} in Section~\ref{sec:space} to show
  the implications \ref{lem:Tr}~--~\ref{lem:CovSup} of the extended
  version are true:

\begin{enumerate}[label={}]

\item \ref{lem:Tr} $|\delta(Tr(t') \setminus Cov( t_{i-1} ))| > 2f$: Follows
  immediately from Lemma~\ref{lem:2f} and
  Definition~\ref{def:notation}.\ref{def:notation:tr}.

\item \ref{lem:Cov}
  $|\delta(~Cov(t') \setminus Cov(t_{i-1})~)| \ge f$: Since by
  Corollary~\ref{cor:Q}, $|Q_i(t_r)|=f$, and by
  Lemma~\ref{lem:facts1}.\ref{obs:ad}, $Q_i(t_r) \subseteq Q_i(t')$,
  we get $|Q_i(t')|=f$.  Hence, by
  Definition~\ref{def:notation}.\ref{def:notation:qi},
  $|\delta(Cov_i(t'))| \geq |Q_i(t')| \geq f$, and by
  Definition~\ref{def:notation}.\ref{def:notation:covi},
  $|\delta(~Cov(t') \setminus Cov(t_{i-1})~)| = |\delta(Cov_i(t'))|
  \geq f$.

\item \ref{lem:CovSup} $Cov(t_i) \supseteq Cov(t_{i-1})$: Follows
  immediately from Definition~\ref{def:ad}.

\end{enumerate}

\end{proof}

\begin{theorem}
  For all $k > 0$, and $f > 0$, let $A$ be an $f$-tolerant
  algorithm emulating a WS-Safe obstruction-free \kreg\/
  using a collection $\mathcal{S}$ of servers. 
  Then, $|\mathcal{S}| \geq 2f+1$.
\label{thm:2f+1}
\end{theorem}

\begin{proof}
  By Lemma~\ref{lem:exhaustive-run}.\ref{lem:Tr}, there exists a run
  $r_1$ of $A$ consisting of $t_1$ steps such that
  $|\delta(Tr(t_1) \setminus Cov(0))| > 2f$.  Therefore,
  $|\mathcal{S}| \geq |\delta(Tr(t_1))| \geq 2f+1$.
\end{proof}

\vspace*{-3mm}
\paragraph{Number of registers per server.}
Theorem~\ref{thm:kf} implies that in case $|\mathcal{S}| = 2f+1$, at
least $(2f+1)k$ registers are required. The following theorem further
refines this result by showing that in this case, each server must
store at least $k$ registers:

\begin{theorem}
  Let $|{\cal S}|=2f+1$. For all $k>0$, $f>0$, every $f$-tolerant
  algorithm emulating an obstruction-free WS-Safe \kreg\/ stores at
  least $k$ registers on each server in ${\cal S}$.
\label{thm:2fplus1}
\end{theorem}

\begin{proof}

  Pick an arbitrary $f>0$, $k>0$, and suppose toward a
  contradiction that there is an $f$-tolerant algorithm $A$
  emulating an obstruction-free WS-Safe \kreg\/ that stores less
  than $k$ registers on some server $s \in \mathcal{S}$ (i.e.,
  $|\delta^{-1}(\{s\})| < k$).
  
  Pick an arbitrary set $F \subset \mathcal{S}$ of size $|F|= f+1$
  such that $s \notin F$.  By Lemma~\ref{lem:exhaustive-run}, there
  exist a sequential write-only run $r_k$ consisting of $k$ high-level
  write invocations $W_1,\dots,W_k$ by $k$ distinct clients, and $k$
  distinct times $t_1 < \dots < t_k$ such that:
  $|\delta(Cov(t_1))| \ge f$ and
  $\delta(Cov(t_1)) \cap F = \emptyset$; and for all
  $i \in [k]\setminus \{1\}$,
  $|\delta(Cov(t_i) \setminus Cov(t_{i-1}))| \ge f$, $Cov(t_i)
  \supseteq Cov(t_{i-1})$, and $\delta(Cov(t_i)) \cap F =
  \emptyset$.
  By induction on $i\in [k]$, it is easy to see that all sets in
  the collection consisting of $Cov(t_1)$ and
  $Cov(t_i) \setminus Cov(t_{i-1})$ where $i \in [k] \setminus \{1\}$
  are pairwise disjoint. Thus, at least $f$ new registers become
  covered at each $t_i$, $i\in [k]$. 
  Moreover, since no registers on
  the servers in $F$ are covered at $t_i$, all registers that
  become covered at $t_i$ must be located on the servers in
  $\mathcal{S} \setminus F$.  
  Therefore, since $|\mathcal{S} \setminus F| = f$ and $s \in
  \mathcal{S} \setminus F$, we conclude that at least $k$
  distinct registers on $s$ must be covered at time 
  $t_k$, that is, $|\delta^{-1}(\{s\}) \cap Cov(t_k)| \geq k$.
  Therefore, $|\delta^{-1}(\{s\})| \geq k$.  
  A contradiction.
\end{proof}

\paragraph{Servers with bounded storage.}
The following result provides a lower bound on the number of
servers for the case the storage available on each server is bounded (as it is
often the case in practice) by a known constant: 

\begin{theorem}
  Let $m>0$ be an upper bound on the number of registers mapped to
  each server in $\mathcal{S}$ \emph{(}i.e.,
  $\forall s \in \mathcal{S}, |\delta^{-1}(\{s\})|\leq m $\emph{)}.
  For all $f > 0$ and $k > 0$, every $f$-tolerant algorithm emulating
  an obstruction-free WS-Safe \kreg\/ from a collection $\mathcal{S}$
  of servers such that $|\mathcal{S}| > 2f + 1$ uses at least
  $\lceil kf/m \rceil + f +1$ servers (i.e.,
  $|\mathcal{S}| \geq \lceil kf/m \rceil + f +1$).
\label{thm:bounded-servers}
\end{theorem}

\begin{proof}
  Fix $F \subset S$ such that $|F|=f + 1$.  
  By Lemma~\ref{lem:exhaustive-run}(a),
  there exists an extension $r_k$ of $r_{k-1}$ ending at time $t_k$
  such that $|\delta(T_r(t_k) \setminus Cov(t_{k-1}))| > 2f$.  Since
  $|F|=f+1$,
  $|\delta(T_r(t_k) \setminus Cov(t_{k-1})) \cap F| \le f + 1$. Hence,
  we receive
\begin{multline*}
|\delta(T_r(t_k) \setminus Cov(t_{k-1})) \setminus F| =\\
|\delta(T_r(t_k) \setminus Cov(t_{k-1}))| - |\delta(T_r(t_k) \setminus
Cov(t_{k-1})) \cap F| \geq 2f+1-f-1 = f
\end{multline*}
Thus,
\begin{equation}
|(T_r(t_k) \setminus Cov(t_{k-1})) \setminus \delta^{-1}(F)| \ge
|\delta(T_r(t_k) \setminus Cov(t_{k-1})) \setminus F| \geq f
\label{eqn:1}
\end{equation}
On the other hand,
\begin{multline}
  |(T_r(t_k) \setminus Cov(t_{k-1})) \setminus \delta^{-1}(F)| = 
  |(T_r(t_k) \setminus Cov(t_{k-1})) \cap \delta^{-1}(\mathcal{S} \setminus F)| = \\
  |(T_r(t_k) \cap \delta^{-1}(\mathcal{S} \setminus F)) \setminus Cov(t_{k-1})| \le  
  |\delta^{-1}(\mathcal{S} \setminus F) \setminus Cov(t_{k-1})|
\label{eqn:2}
\end{multline}

\noindent  
Since by Lemma~\ref{lem:exhaustive-run}(b) and (d), $|Cov(t_{k-1}|\ge (k-1)f$ and
$\delta^{-1}(\mathcal{S} \setminus F) \supseteq Cov(t_{k-1})$, we get
\begin{equation}
|\delta^{-1}(\mathcal{S} \setminus F) \setminus Cov(t_{k-1})| =
|\delta^{-1}(\mathcal{S} \setminus F)| - |Cov(t_{k-1})| \le
(|\mathcal{S} \setminus F|)m - (k-1)f
\label{eqn:3}
\end{equation}

\noindent
Combining (\ref{eqn:3}) with (\ref{eqn:1}) and (\ref{eqn:2}), we get
\begin{equation*}
  (|\mathcal{S} \setminus F|)m - (k-1)f \ge 
|(T_r(t_k) \setminus Cov(t_{k-1})) \setminus \delta^{-1}(F)| \ge f
\end{equation*}
Since $\mathcal{S} \supset F$ and $|F|=f+1$, we obtain
$(|\mathcal{S} \setminus F|)m - (k-1)f = |\mathcal{S}|m - (f+1)m -
(k-1)f \ge f$.
Therefore, $|\mathcal{S}|m \ge fm + m + kf$, which implies that
$|\mathcal{S}| \ge kf / m + f + 1$. Since $|\mathcal{S}|$ is an
integer, we conclude that
$|\mathcal{S}| \ge \lceil kf / m \rceil + f + 1$.
\end{proof}

\paragraph{Adaptivity to Contention}  Given a run fragment
 $r$ of an emulation algorithm, the {\em point
   contention}~\cite{point-cont-afek,point-cont} of $r$,
   $\PntCont(r)$, is the maximum number of clients that have an
   incomplete high-level invocation after some finite prefix of
   $r$. Similarly, we use $\PntCont(op)$ to denote
   $\PntCont(r_{op})$, where $r_{op}$ is the run fragment
   including all events between the $op$'s invocation and response.
 
 The resource complexity of $A$ is {\em adaptive to point
 contention}\/ if there exists a function $M$ such that after
 all finite runs $r$ of $A$, the resource consumption of $A$ in
 $r$ is bounded by $M(\PntCont(r))$. Likewise, the time
 complexity of $A$ is {\em adaptive to point contention}\/ if
 there exists a function $T$ such that for each client $c_i$,
 and operation $op$, the time to complete the invocation of $op$
 by $c_i$ is bounded by $T(\PntCont(op))$.\\

\noindent We show that no WS-Safe
obstruction-free MWSR register can have a fault-tolerant
emulation adaptive to point contention:

\begin{theorem}
\label{theorem:adaptive}

  For any $f>0$, there is no $f$-tolerant algorithm that
  emulates an WS-Safe obstruction-free \kreg\ with resource
  complexity adaptive to point contention.

\end{theorem}

%

\begin{proof}
  By Lemma~\ref{lem:exhaustive-run}, there exists a run
  $r$ of $A$ consisting of $k$ high-level writes by $k$
  distinct clients such that the resource complexity
  grows by $f$ for each consecutive write that completes
  in $r$ whereas the point contention remains equal $1$
  for the entire $r$. 
  We conclude that no function mapping point contention
  to resource consumption can exist, and therefore, $A$'s
  resource complexity is not adaptive to point
  contention. 
\end{proof}

\newpage


\section{Upper Bound Algorithm}
\label{app:upper}

In this section, we will give a detailed description of our upper
bound construction discussed in Section~\ref{sub:UpperBound} along
with the correctness proof.

\begin{algorithm}[H]

\caption{$\forall f>0 ~\forall k>0, ~\forall n =
|\mathcal{S}| \geq 2f+1$.}
\begin{algorithmic}[0]
\small

\Types{}{}
\State $TSVal = \mathbb{N} \times \mathbb{V}$, with selectors
$ts$ and $val$.
\State $States =  TSVal \times
2^{TSVal} \times 2^{\mathcal{B}} \times 2^{\mathcal{B}}$ 
 with selectors \emph{tsVal},  \emph{rdSet},
\emph{wrSet} and \emph{coverSet}.
\EndTypes

\BaseObjectsServers{}{}
\State  $\forall b \in \delta^{-1}(\mathcal{S})$,
$b \in TSVal$, initially,  $\langle 0,
v_0 \rangle$.
\State Let $z \triangleq \big\lfloor \frac{n- (f+1)}{f}
\big\rfloor$, $y \triangleq zf + f +1$, and $m \triangleq
\big\lceil \frac{k}{z} \big\rceil$.
\State $R = \{ R_0,\ldots,R_{m-1}\}  \subset
2^{\delta^{-1}(\mathcal{S})}$ s.t.\:
   \State \hspace*{4mm} 1. $\forall i \in \{0,\ldots,m-2\}$,
   $|R_i| = y$.
   If $ \big\lceil \frac{k}{z} \big\rceil = 
   \big\lfloor \frac{k}{z} \big\rfloor$, then $|R_{m-1}| = y$.
   Else, $|R_{m-1}| = (k -\big\lfloor \frac{k}{z} \big\rfloor z)
   f + f +1$.
   \State \hspace*{4mm} 2.  $\forall R_i,R_j \in R$, $R_1 \cap
   R_j = \emptyset$.
   \State \hspace*{4mm} 3.  $\forall R_i\in R$, $|\delta(R_i)| =
   |R_i|$.
   
\EndBaseObjectsServerss


\States{}{}

\State $\forall i \in [k], ~State_i \in States$,
initially, 
\State $\langle  \langle \langle0,
1 \rangle, v_0 \rangle, \emptyset, R_j, \emptyset
\rangle$, where $j = \big\lfloor \frac{i}{z} \big\rfloor$.

\EndStates

 \Statex

\Code{for client $c_i$, $1\le i \le k$:}
\EndCode
\end{algorithmic}

\vspace*{-0.4cm}%
\begin{multicols}{2}
\label{alg:upper}
\begin{algorithmic}[1]

\small%

\Operation{$\Write(v)$}{} %
	
	\State $value \gets collect()$
	\label{line:writeCollect}
	\State $State_i.tsVal.val \gets v$
	\State $State_i.tsVal.ts \gets  value.ts+1 $
	\label{line:pickTs}

 	\State $j \gets \big\lfloor \frac{i}{z} \big\rfloor$
 	\label{line:row}
	\Statex \Comment{do not handle responds between lines
	\ref{line:cover} to \ref{line:trigger}} 
	\State $State_i.coverSet \gets R_j \setminus State_i.wrSet$
	\label{line:cover} 
	\State $State_i.wrSet \gets \emptyset$
	\State \textbf{$||$ for each} $b \in R_j$
	\label{line:forTrigger}
	\State \hspace*{0.4cm} \textbf{if} $b \notin
	State_i.coverSet$
	\label{line:ifCovered}
	\State \hspace*{0.8cm} trigger $b.\lwrite(State_i.tsVal)$
	\label{line:trigger}
	
	\State \textbf{wait until} $|State_i.wrSet| \geq |R_j|-f$
	\label{line:writeWait}
	\State \Return $ack$
\EndOperation %

\Statex

\Macro{$scan(s)$}{}
	\For{\textbf{each} $b \in \delta^{-1}(s)$} 
		\State trigger $b.\lread()$
		\State wait for the matching response 

\EndFor
\EndMacro

\Statex 
\Statex 
\Statex

\Operation{$\Read()$}{}
	\State $value \gets collect()$
	\State \Return value.val 
\EndOperation

\Statex 

\Macro{$collect()$}{}
	\State $State_i.rdSet \gets \emptyset$
	\State \textbf{$||$ for each} $s \in \mathcal{S}$
	\label{line:collectStart}
	\textbf{do}
		\State \hspace*{0.4cm} $scan(s)$
		\State wait for $n-f$ \emph{scans} to complete
		\label{line:waitScans}
		\State $ts \gets  max(\{ts' \mid \langle ts', * \rangle \in
		State_i.rdSet \})$
		\State \Return $\langle v, ts'\rangle \in State_i.rdSet$ : $ts' = ts$
		\label{line:collectEnd}

\EndMacro

	\Statex 

\RRespond{$res$}

	\State $State_i.rdSet \gets State_i.rdSet \cup \{res\}$ 

\EndRRespond

\Statex

\WRespond{}
\label{line:wrespnodBegin}

	\If{$b \in State_i.coverSet$}
	
		\State $State_i.coverSet \gets State_i.coverSet \setminus
		\{b\}$
		\State trigger $b.\lwrite(State_i.tsVal)$
		
		
	\Else
		\State $State_i.wrSet \gets State_i.wrSet \cup \{b\}$

	\EndIf

\EndWRespond
\label{line:wrespnodEnd}





\end{algorithmic}
\end{multicols}
\end{algorithm}

The registers store values in $\mathbb{V}$ each of which is
associated with a unique timestamp.
(Note that since safety is required only in write-sequential
runs, we do not need to break ties with clients' ids.)
To write a value $v$ to the emulated register, a client $c_i$
first accesses a read quorum (via collect() in lines
\ref{line:collectStart}--\ref{line:collectEnd}) and selects a
new timestamp $ts$ which is higher than any other timestamp that
has been returned.
It then proceeds to trigger low-level writes of $\tup{ts}{v}$ on
registers in $R_j = \big\lfloor \frac{i}{z} \big\rfloor$, so as
to ensure that (1) $\tup{ts}{v}$ is stored in a write quorum
$wq$ (lines~\ref{line:forTrigger}--\ref{line:writeWait}), and
(2) no more than $f$ registers in $R_j$ are left covered by
$c_i$'s writes (the current and the previous operations). 
The latter is achieved by preventing $c_i$ from triggering a new
low-level write on every register that has not yet responded to
the previously triggered one
(lines~\ref{line:ifCovered}--\ref{line:trigger}).
To read a value, a client simply reads all registers in a read quorum,
via collect(), and returns the value having the highest timestamp.


The space complexity of the algorithm is $\Sigma_{R_i \in R}
|R_i| = \big\lfloor \frac{k}{z} \big\rfloor y + (k - \big\lfloor
\frac{k}{z} \big\rfloor z)f  + (f+1)(\big\lceil
\frac{k}{z} \big\rceil - \big\lfloor \frac{k}{z}
\big\rfloor)= \cdots = kf + \big\lfloor \frac{k}{z}
\big\rfloor (f+1) = kf +  \big\lceil \frac{k}{\big\lfloor \frac{x-(f+1)}{f}
\big\rfloor} \big\rceil$ registers.
Below, we prove that the algorithm satisfies wait-freedom and
write-sequential regularity.  The following observation follows
from code and the construction of the sets in $R$; (1) writers
never trigger low-level writes on base object with pending
low-lever writes from previous \Write s, (2) writers wait for
$n-f$ base objects to reply (line~\ref{line:waitScans}), and for
every set $R_i \in R$, the number of client that write
to registers in $R_i$ is $\big\lfloor \frac{|R_i|-(f+1)}{f}
\big\rfloor = \frac{|R_i|-(f+1)}{f}$.

\begin{observation}
\label{obs:covering}

For every $0 < i \leq k$ for every time $t$ in a run $r$, if
writer $c_i$ have no pending \Write\ at $t$ then it covers at
most $f$ base objects at time $t$.

\end{observation}

\begin{lemma}
\label{lemma:writeTS}

Consider a write-sequential run $r$ of the algorithm, and
consider two sequential \Write s $W_i,W_j$ in $r$ s.t.\ $W_i$
precedes $W_j$. 
Then $W_j$'s value is associated with a bigger timestamp than
$W_i$'s value.

\end{lemma}

\begin{proof}

Since $W_i$ precedes $W_j$, $W_j$ starts the $collect$ in
line~\ref{line:writeCollect} after $W_i$ returns.
$W_i$ triggers low level writes with its value and timestamp on
base objects in $R_l$ ($l = \big\lfloor \frac{i}{z}
\big\rfloor$) that are not covered by its
previous \Write s, and waits for $|R_l|-f$ low level writes to
respond (line~\ref{line:writeWait}) before it returns.
Thus, since $|\delta(R_l)| = |R_l|$, $W_j$ starts its $collect$
after $W_i$ writes its timestamp to at least $|R_l|-f $ base
objects in different servers, none of which is covered by
low-level write of $W_i$ previous \Write s.

Moreover, since the number of writers excluding $W_i$ that
write to base objects in $R_j$ is $\big\lfloor
\frac{|R_l|-(f+1)}{f} \big\rfloor -1 = \frac{|R_i|-(f+1)}{f} -1$,
readers do not write, and each writer covers at most
$f$ base objects,  we get that at least $f+1$ servers has a base
object that stores $W_i$'s timestamp when $W_j$ begins its $collect$.
Now since $collect$ reads all base object in at least
$n-f$ servers (line~\ref{line:waitScans}), $W_j$ sees $W_i$'s
timestamp and picks a bigger one (line~\ref{line:pickTs}).


\end{proof}

%
%
%

\begin{corollary}
\label{cor:writeTS}

Consider a write-sequential run $r$ of the algorithm. 
If \Write\ $W_i$ precedes \Write\ $W_j$, then $W_j$ is associated
with a bigger timestamp than $W_i$.

\end{corollary}

\begin{lemma}
\label{lemma:readTS}

Consider a write-sequential run $r$ of the algorithm, and
a read operation $rd$ and a \Write\ $W$ in $r$.
Let $ts$ be the timestamp associated with $W$.
If $W$ precedes $rd$, than $rd$ returns a value associated with
timestamp $ts' \geq ts$. 

\end{lemma}

\begin{proof}

Let $t$ be the time when $W$ returns, and assume w.l.o.g that
$W$ is performed by client $c_i$ s.t.\  $ \big\lfloor
\frac{i}{z} \big\rfloor = j$.
Before $W$ returns it $c_i$ triggers low level writes with
its value and timestamp on base objects in $R_j$
that are not covered by its previous \Write s, and waits for
$|R_j|-f $ low level writes to respond.
The number of clients excluding $c_i$ that trigger low-level
writes on base objects is $R_j$ is $\big\lfloor
\frac{|R_l|-(f+1)}{f} \big\rfloor -1 = \frac{|R_i|-(f+1)}{f}
-1$, and by Observation~\ref{obs:covering}, each client covers
at most $f$ base objects at time $t$.
By Corollary~\ref{cor:writeTS} and since readers do not write,
every low level write in $r$ that is triggered after time $t$ is
associated with a bigger timestamp than $ts$. 
Therefore, since $|\delta(R_l)| = |R_l|$, there is a set of
$f+1$ base objects, each of which mapped to a different server,
s.t.\ at any time $t' \geq t$ the timestamp each of them stores
is bigger than or equal to $ts$.

Since $W$ precedes $rd$, $rd$ starts the $collect$ after time
$t$.
And since $collect$ reads all base object in at least
$n-f $ servers, $rd$ sees at least one value associated
with timestamp bigger than or equal to $ts$, and thus, returns a
value associated with timestamp $ts' \geq ts$.

\end{proof}

\begin{definition}

For every write-sequential run $r$, for every \emph{read} $rd$
in $r$ that returns a value associated with timestamp $ts$ we
define the sequential run $\sigma_{r_{rd}}$ as follows:
All the completed write operations in $r$ are ordered in
$\sigma_{r_{rd}}$ by their timestamp, and $rd$ is added after
the \Write\ operation that is associated with $ts$.

\end{definition}

\noindent In order to show that the algorithm simulates a
write-sequential regular register we need to proof that for
every write-sequential run $r$, for every \emph{read} $rd$,
$\sigma_{r_{rd}}$ preserves the real time order of $r$ and the
sequential specification.
Note that the sequential specification is satisfied by
construction, and we prove the real time order in the next lemma.

\begin{lemma}
\label{lem:realTime}

For every write-sequential run $r$, for every complete read $rd$ 
that returns a value associated with timestamp $ts$ in $r$,
$\sigma_{r_{rd}}$ preserves $r$'s operation precedence relation
(real time order).

\end{lemma}

\begin{proof}

By Corollary~\ref{cor:writeTS}, the real time order of $r$
between every two \Write\ operations is preserved in
$\sigma_{r_{rd}}$.
We left to show that the real time order of $r$ 
between $rd$ and any \Write\ $W$ in $\sigma_{r_{rd}}$ is
preserved.
Consider two cases: 

\begin{itemize}

  \item $W$ precedes $rd$ in $r$. By Lemma~\ref{lemma:readTS},
  $W$ is associated with a timestamp smaller than or equal to
  $ts$, and thus, by construction of $\sigma_{r_{rd}}$ the real
  time order between $rd$ and $W$ is preserved.
  
  \item $rd$ precedes $W$ in $r$. Let $W_{ts}$ be the \Write\
  operation associated with timestamp $ts$. Since $rd$ returns a
  value associated with timestamp $ts$, $W_{ts}$ starts before
  $rd$ completes, and since $r$ is write-sequential, $w_{ts}$
  precedes $W$ in $r$. Thus, by lemma~\ref{lemma:writeTS},
  $W$ is associated with bigger timestamp than $ts$.
  Therefore, by construction of $\sigma_{r_{rd}}$ the real
  time order between $rd$ and $W$ is preserved.
   
\end{itemize}

\end{proof}

\begin{corollary}
\label{cor:linearization}

For every write-sequential run $r$, for every complete read $rd$ 
in $r$, there is a linearization of $rd$ and all the \Write\
operations in $r$.

\end{corollary}

\begin{thmclone}{thm:upper}~ 
For all $k>0$, $f>0$, and $n > 2f$, there exists an $f$-tolerant
algorithm emulating a wait-free WS-Regular \kreg\/ using a collection
of $n$ servers storing $kf + \big\lceil \frac{k}{z}
\big\rceil (f+1)$ wait-free $z$-writer/multi-reader atomic
base registers where $z=\big\lfloor \frac{n-(f+1)}{f} \big\rfloor$.
\end{thmclone}


\begin{proof}

  By Corollary~\ref{cor:linearization}, the code in
  Algorithm~\ref{alg:upper} satisfies WS-regularity.  Now notice that
  in both \Write\ and \Read\ operations clients never wait for more
  than $n-f$ servers to respond, and thus, wait-freedom
  follows. We conclude that Algorithm~\ref{alg:upper} satisfies the
  theorem.

\end{proof}

\end{document}